\documentclass[12pt, draftclsnofoot, onecolumn]{IEEEtran}
\usepackage{amsfonts}
\usepackage{amssymb}
\usepackage{amsthm}
\usepackage{amsmath,amsfonts,amssymb}
\usepackage[dvips]{graphicx}
\usepackage{verbatim}
\usepackage{setspace}
\usepackage{bm}
\usepackage{algorithmic} 
\usepackage[ruled,vlined]{algorithm2e}
\usepackage{cite}

\usepackage{changepage}
\usepackage{pdfpages}
\usepackage{color}
\newtheorem{theorem}{Theorem}
\newtheorem{lemma}{Lemma}
\newtheorem{proposition}{Proposition}

\newcommand{\figwidth}{8}
\IEEEoverridecommandlockouts

\begin{document}
%
\title{Millimeter-Wave NOMA with User Grouping, Power Allocation and Hybrid Beamforming}

%
%
%
\author{Lipeng Zhu, ~\IEEEmembership{Student Member,~IEEE,}
        Jun Zhang,
        Zhenyu Xiao,~\IEEEmembership{Senior Member,~IEEE,}
        Xianbin Cao,~\IEEEmembership{Senior Member,~IEEE,}
        Dapeng Oliver Wu,~\IEEEmembership{Fellow,~IEEE}
        and Xiang-Gen Xia,~\IEEEmembership{Fellow,~IEEE}
\thanks{L. Zhu, Z. Xiao and X. Cao are with the School of
Electronic and Information Engineering, Beihang University, Beijing 100191, China. \{zhulipeng@buaa.edu.cn, xiaozy@buaa.edu.cn, xbcao@buaa.edu.cn\}.}
\thanks{J. Zhang is with the Advanced Research Institute of Multidisciplinary Science, Beijing Institute of Technology, Beijing, 100081, China. \{buaazhangjun@vip.sina.com\}.}
\thanks{D. O. Wu is with the Department of Electrical and Computer Engineering, University of Florida, Gainesville, FL 32611, USA. \{dpwu@ufl.edu\}.}
\thanks{X.-G. Xia is with the Department of Electrical and Computer Engineering, University of Delaware, Newark, DE 19716, USA. \{xianggen@udel.edu\}.}
}

%
%

\maketitle

\begin{abstract}
This paper investigates the application of non-orthogonal multiple access in millimeter-Wave communications (mmWave-NOMA). Particularly, we consider downlink transmission with a hybrid beamforming structure. A user grouping algorithm is first proposed according to the channel correlations of the users. Whereafter, a joint hybrid beamforming and power allocation problem is formulated to maximize the achievable sum rate, subject to a minimum rate constraint for each user. To solve this non-convex problem with high-dimensional variables, we first obtain the solution of power allocation under arbitrary fixed hybrid beamforming, which is divided into intra-group power allocation and inter-group power allocation. Then, given arbitrary fixed analog beamforming, we utilize the approximate zero-forcing method to design the digital beamforming to minimize the inter-group interference. Finally, the analog beamforming problem with the constant-modulus constraint is solved with a proposed boundary-compressed particle swarm optimization algorithm. Simulation results show that the proposed joint approach, including user grouping, hybrid beamforming and power allocation, outperforms the state-of-the-art schemes and the conventional mmWave orthogonal multiple access system in terms of achievable sum rate and energy efficiency.
\end{abstract}


%
\IEEEpeerreviewmaketitle

\section{Introduction}
\IEEEPARstart{M}{illimeter}-wave (mmWave) communication has been proposed as one of the candidate key techniques for the fifth-generation (5G) wireless communications and beyond \cite{andrews2014will,niu2015survey,rapp2013mmIEEEAccess,XiaoM2017survmmWave}. The abundant spectrum (30-300GHz) in mmWave-band can provide great potentials to meet the requirements of high data rates and low transmission latency. Due to the high path loss, large antenna array is usually utilized in mmWave communications, where beamforming techniques are required to increase the spectrum efficiency \cite{XiaoM2017survmmWave,xiao2017mmWaveFD,andrews2016modeling}. Fully digital beamforming (DBF) is one of the signal processing approaches in baseband \cite{Rusek2013MIMO,Gao2017MIMO}, where each antenna is driven by an independent radio frequency (RF) chain, and multiple data streams can be transmitted simultaneously. However, the DBF architecture results in unaffordable hardware cost and energy consumption in the mmWave-band with large antenna array \cite{Daill2017}. In contrast, analog beamforming (ABF), where the antennas share only one RF chain, is an energy-efficient alternative \cite{Ding2017random,xiao2016codebook}. However, one RF chain can support only one data stream in general, which limits the spectrum efficiency. In consideration of the compromission between energy efficiency and spectrum efficiency, hybrid analog and digital beamforming (HBF) was proposed and preferred \cite{Gao2016hyb,Daill2017,Dai2018MIMONOMA}. With a small number of RF chains connected to a large number of antennas, beam gain and interference management can be achieved simultaneously.

One of the typical application scenarios for the 5G wireless communications is the massive connectivity. However, for mmWave communications with the conventional orthogonal multiple access (OMA) schemes, such as time division multiple access (TDMA), code division multiple access (CDMA), orthogonal frequency division multiple access (OFDMA), and space division multiple address (SDMA), the number of the users for each data stream in the same time-frequency-code-space resource block (RB) is one \cite{Daill2017,Dai2018MIMONOMA,xiao2018mmWaveNOMA,Zhu2018UplinkNOMA}. Thus, the total number of served users is limited, which is no more than the number of RF chains in each RB \cite{Daill2017,Dai2018MIMONOMA,xiao2018mmWaveNOMA,Zhu2018UplinkNOMA}. To address this problem, non-orthogonal multiple access (NOMA) was proposed to combine with mmWave communications \cite{Ding2017random,Daill2017,Ding2017survNOMA,Dai2018NOMAsurvey}. In contrast to the conventional OMA schemes, NOMA can transmit the signals for different users in the same RB, while distinguishing them in the power domain. By employing superposition coding at the transmitter and successive interference cancellation (SIC) at the receiver, the users with different channel conditions can be served simultaneously. The number of served users in the same time-frequency-code-space RB can be improved manyfold \cite{Benjebbour2013ConceptNOMA,Dai2015NOMA5G,Ding2017survNOMA,Dai2018NOMAsurvey,Choi2014NOMA,Zhu2018NOMAPSO}. Note that the implementation of NOMA does not result in extra delay caused by channel estimation and feedback compared with OMA \cite{sun2018delay}. Although the SIC at the receiver brings in supererogatory computation for demodulation and decoding at the NOMA user, the corresponding latency in the physical layer is negligible compared with the delay in the network layer. The performance analysis of NOMA for Ultra-Reliable and Low-Latency Communications (URLLC) has been investigated in \cite{Amjad2018delay}, where grant-free NOMA with short-packet communications has significantly reduced the latency and improved the reliability for URLLC to support the time-critical applications. Moreover, it has been verified that NOMA with short-packet communications can significantly outperform OMA by achieving a higher effective throughput with the same latency requirement \cite{sun2018delay}. Besides, several schemes have been proposed to realize the tradeoff between the capacity (or energy efficiency) and the delay \cite{Choi2017delay,Ding2018delay,Ning2019delay}. The analysis and optimization of the delay for NOMA are beyond the scope of this paper.

It has been verified that applying NOMA in mmWave communications (mmWave-NOMA) can significantly improve the throughput capacity compared with mmWave-OMA \cite{Ding2017random,Daill2017,Dai2018MIMONOMA,xiao2018mmWaveNOMA,Zhu2018UplinkNOMA}. Due to the directional feature of mmWave transmission, it is ideal for the users whose channels are highly correlated to perform NOMA. There are several prior works on mmWave-NOMA with ABF. Using random ABF, mmWave-NOMA could outperform mmWave-OMA in terms of outage sum rates, respecting to a targeted data rate of the strong user \cite{Ding2017random}. In \cite{xiao2018mmWaveNOMA}, a 2-user downlink mmWave-NOMA scenario with ABF was considered. A joint Tx beamforming and power allocation problem was formulated and solved to maximize the achievable sum rate (ASR), subject to a minimum rate constraint for each user. In \cite{Zhu2018UplinkNOMA}, a joint Rx beamforming and power control problem was solved in a 2-user uplink mmWave-NOMA system. Furthermore, a joint Tx-Rx beamforming and power allocation problem was solved for $K$-user downlink mmWave-NOMA in \cite{Zhu2018NOMAPSO}. The closed-form optimal power allocation and Rx beamforming were obtained under arbitrary fixed Tx beamforming, and a boundary-compressed particle swarm optimization (BC-PSO) algorithm was proposed to solve the ABF problem with the constant modulus (CM) constraint.

In addition, mmWave-NOMA with HBF was also investigated in several literatures. In \cite{Daill2017}, a new transmission scheme of beamspace multiple-input multiple-output NOMA (MIMO-NOMA) was proposed, where the number of users can be larger than the number of RF chains. Based on the equivalent-channel hybrid precoding scheme, an iterative algorithm was developed to obtain the optimal power allocation for the users. In \cite{Dai2018MIMONOMA}, a user grouping algorithm and an HBF algorithm were proposed for mmWave-MIMO-NOMA system with simultaneous wireless information and power transfer. Then, the optimization for power allocation and power splitting factors was operated to maximize the ASR. The optimal power allocation and user scheduling were obtained with the branch and bound approach in \cite{Cui2018mmWaveNOMA}, where HBF is random and fixed. In \cite{Wu2017hybridBF}, the authors considered the problems of user pairing, hybrid beamforming and power allocation separately in an mmWave-NOMA system. In \cite{Zhang2017mmWaveMIMONOMA}, a capacity analysis for the integrated NOMA-mmWave-massive-MIMO systems was provided based on a simplified mmWave channel model. In \cite{Wei2018NOMA}, a multi-beam NOMA framework for hybrid mmWave systems was proposed, where a beam splitting technique was introduced to generate multiple analog beams to facilitate the NOMA transmission.

In this paper, we investigate mmWave-NOMA with HBF structures. Different from the works above, we consider user grouping and jointly optimize HBF and power allocation. Particularly, we consider a single-cell downlink system, where the base station (BS) is equipped with a large antenna array, and serves multiple single-antenna users. The contributions of this paper are summarized as follows\footnote{In our previous work \cite{Zhu2018NOMAPSO}, a mmWave-NOMA system with the pure analog beamforming structure was considered, where several key problems for HBF structure were not included, e.g., the user grouping, digital beamforming, inter-group interference suppression and power allocation among different groups, which bring new challenges for the multi-group mmWave-NOMA system.}.
\begin{enumerate}
  \item To implement NOMA in mmWave communications with HBF, we propose a user grouping algorithm first, where K-means algorithm is utilized and the normalized channel correlation is defined as the measure. The users with high channel correlation are assigned to the same group, while the users with low channel correlation are assigned to different groups, which can significantly mitigate the interference between different groups of users. Then, a problem jointly optimizing power allocation and HBF is formulated to maximize the ASR of the users, subject to a minimum rate constraint for each user.
  \item We obtain a sub-optimal solution of the power allocation problem under arbitrary and fixed HBF. Since the power allocation problem is non-convex, we divide it into two sub-problems, i.e., intra-group power allocation (intra-GPA) and inter-group power allocation (inter-GPA). Significantly, we prove the proposed solution of power allocation is globally optimal under ideal beam pattern (i.e., no inter-group interference).
  \item We design the HBF matrix to suppress the inter-group interference as well as maximize the ASR. In the proposed solution, DBF is designed by using the approximate zero-forcing (AZF) method under arbitrary and fixed ABF. Then, substituting the obtained power allocation and DBF as the function of the ABF matrix, we utilize the boundary-compressed particle swarm optimization (BC-PSO) algorithm to solve the ABF problem, which realizes the joint optimization of power allocation and HBF.
  \item We evaluate the performance of the proposed user grouping, power allocation and HBF algorithm for mmWave-NOMA through simulations. The simulation results show that the proposed solution is significantly better than those of state-of-the-art schemes and the conventional mmWave-OMA system in terms of ASR. The energy-efficiency (EE) performance of the proposed mmWave-NOMA scheme with an HBF structure outperforms the fully digital MIMO structure\footnote{In the simulation, we compare the ASR/EE performance of mmWave-NOMA with the scheme proposed in \cite{Dai2018MIMONOMA}, which is regarded as the benchmark.}. The ASR of the proposed solution is close to the ideal case with no inter-group interference, which demonstrates that the designed HBF can significantly achieve low inter-group interference.
\end{enumerate}

The rest of the paper is organized as follows. In Section II, we present the system model. In Section III, we first propose the user grouping algorithm and formulate the problem. Then, we provide a solution of power allocation with an arbitrary fixed HBF in Section IV. In Section V, we design DBF and ABF. In Section VII, we summarize the complete solution and provide the computational complexity. Simulation results are given to demonstrate the performance of the proposed solution in Section VII, and the paper is concluded finally in Section VIII.

Symbol Notation: $a$, $\mathbf{a}$, $\mathbf{A}$ and $\mathcal{A}$ denote a scalar, a vector, a matrix and a set, respectively. $(\cdot)^{\rm{T}}$, $(\cdot)^{\rm{H}}$ and $(\cdot)^{\dag}$ denote transpose, conjugate transpose and pseudo inverse, respectively. $|a|$ and $\|\mathbf{a}\|$ denote the absolute value of $a$ and Frobenius norm of $\mathbf{a}$, respectively, while $|\mathcal{A}|$ denotes the number of elements in set $\mathcal{A}$. $\mathbb{E}(\cdot)$ denotes the expectation operation. $[\mathbf{a}]_i$, $[\mathbf{A}]_{i,:}$, $[\mathbf{A}]_{:,j}$ and $[\mathbf{A}]_{i,j}$ denote the $i$th entry of $\mathbf{a}$, the $i$th row, the $j$th column, and the entry in the $i$th row and the $j$th column of $\mathbf{A}$, respectively. $\mathbf{I}_K$ is the $K \times K$ identity matrix and $\Phi$ denotes the empty set.

\section{System Model}
\subsection{System model}
\begin{figure}[t]
\begin{center}
  \includegraphics[width=8.8 cm]{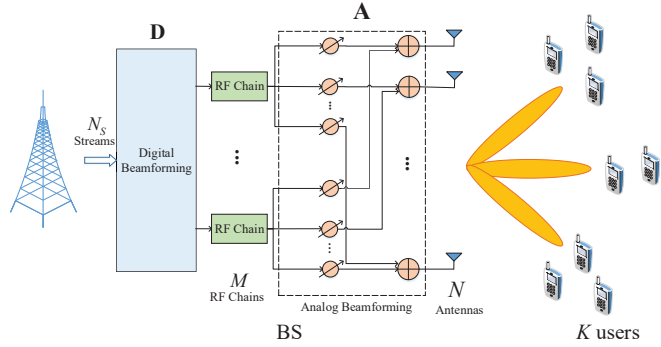}
  \caption{Illustration of the architecture of the BS, which is equipped with $M$ RF chains and $N$ antennas.}
  \label{fig:system}
\end{center}
\end{figure}
In this paper, we consider a single-cell downlink mmWave-NOMA system. The BS is equipped with HBF structure, where $N$ antennas share $M$ RF chains. $K$ single-antenna users are served simultaneously, where $K>M$. The architecture of the BS is shown in Fig. \ref{fig:system}, which is a fully connected HBF structure\footnote{It is worthy of noting that the proposed approach in this paper can also be directly used for the partially connected HBF structure \cite{Dai2018MIMONOMA}.}. $N_{S}$ data streams in the baseband are precoded by the DBF matrix $\mathbf{D}^{M\times N_{S}}$. After passing through the corresponding RF chain, the digital-domain signal from each RF chain is delivered to $N$ phase shifters (PSs) to perform ABF. Thus, the ABF matrix is $\mathbf{A}^{N\times M}$.

In order to achieve a higher multiplexing gain, the number of data streams is assumed to be equal to the number of RF chains in this paper, i.e., $N_{S}=M$. Thus, the $K$ users should be first scheduled into $M$ groups, and each group is corresponding to an independent data stream. The users in the same group can perform NOMA and implement SIC, while the signals from different groups of users are treated as interference. The details of user grouping will be shown later. Denote the user set of the $m$th group as $\mathcal{G}_{m}$. As a result, we have $\mathcal{G}_{i}\cap \mathcal{G}_{j}=\Phi$ for $i\neq j$ and $\sum \limits_{m=1}^{M} |\mathcal{G}_{m}|=K$, where $|\mathcal{G}_{m}|$ denotes the number of users for $\mathcal{G}_{m}$. Since $M$ RF chains can support $M$ data streams at most, there should be at least one user in each group to avoid the idleness of the RF resource, and thus we have $|\mathcal{G}_{m}|\geq 1$. Then, the received signal for the $n$th user in the $m$th group is
\begin{equation}
y_{m,n}=\mathbf{h}_{m,n}^{\rm{H}}\mathbf{ADPs}+u_{m,n},
\end{equation}
where $\mathbf{h}_{m,n}$ with $N \times 1$ dimension is the channel response vector between the BS and the $n$th user in the $m$th group. $u_{m,n}$ is the Gaussian white noise at the user with average power $\sigma^2$. $\mathbf{s}^{K \times 1}$ is the vector of the transmission signals, where $\mathbf{s}=[s_{1,1}, \cdots,s_{1,|\mathcal{G}_{1}|},\cdots,s_{M,1}, \cdots,s_{M,|\mathcal{G}_{M}|}]^{\mathrm{T}}$ and $E(\mathbf{ss}^{\rm{T}})=\mathbf{I}_{K}$, and $\mathbf{P}$ is the $M \times K$ power allocation matrix: $\mathbf{P}=\mathrm{diag}\{\mathbf{p}_1,\mathbf{p}_2,\cdots,\mathbf{p}_M\}$ and $\mathbf{p}_m=[\sqrt{p_{m,1}},\sqrt{p_{m,2}},\cdots,\sqrt{p_{m,|\mathcal{G}_{m}|}}]$. $\mathbf{D}$ is the DBF matrix. $\mathbf{A}$ is the ABF matrix with the CM constraint of \cite{xiao2018mmWaveNOMA,Zhu2018UplinkNOMA,Zhu2018NOMAPSO}
\begin{equation}
|[\mathbf{A}]_{i,j}|=\frac{1}{\sqrt{N}},~1\leq i \leq N,~1\leq j \leq M.
\end{equation}

We define the HBF matrix as
\begin{equation}
\mathbf{W}=\mathbf{AD}=[\mathbf{w}_{1},\mathbf{w}_{2},\cdots,\mathbf{w}_{M}].
\end{equation}
Since we separate the transmission power from HBF, it is without loss of generality to assume that each column of the HBF matrix has a unit norm, i.e.,
\begin{equation}
\|\mathbf{w}_{m}\|=1,~1\leq m \leq M.
\end{equation}

Subject to limited scattering in the mmWave band, multipath is mainly caused by reflection. As the number of the multipath components (MPCs) is small in general, the mmWave channel has directionality and appears spatially sparse in the angle domain \cite{peng2015enhanced,wang2015multi,Lee2014exploiting,Gao2016ChannelEst,xiao2016codebook,alkhateeb2014channel}. Different MPCs have different angles of departure (AoDs) and angles of arrival (AoAs). Without loss of generality, we adopt the directional mmWave channel model assuming a uniform linear array (ULA) with a half-wavelength antenna spacing. For the $N\times 1$ channel response vector $\mathbf{h}_{m,n}$, we adopt the widely used Saleh-Valenzuela channel for mmWave communications \cite{Ding2017random,Daill2017,Dai2018MIMONOMA}, which is\footnote{Since we concentrate on the user grouping and resource allocation for mmWave-NOMA, the channel estimation problem is beyond the scope of this paper. We assume that the channel state information (CSI) between the BS and the users is known by the BS. A number of approaches on mmWave channel estimation have been proposed and could be referred, such as, \cite{xiao2016codebook,Kokshoorn2017ChannelEstim,xiao2018codebook,Hu2018channelEstim}.}
\begin{equation} \label{eq_oriChannel}
\mathbf{h}_{m,n}=\sum_{\ell=1}^{L_{m,n}}\lambda_{m,n}^{(\ell)}\mathbf{a}(N,\theta_{m,n}^{(\ell)}).
\end{equation}
Note that for convenience, we denote the channel coefficients in terms of both the indexes $m$ and $n$ in \eqref{eq_oriChannel}, where $m~(1\leq m \leq M)$ represents the $m$th group, and the index $n~(1\leq n \leq |\mathcal{G}_{m}|)$ represents the $n$th user in each group. $\lambda_{m,n}^{(\ell)}$ is the complex coefficient of the $\ell$-th MPC of the channel response vector for the $n$th user in the $m$th group. $\theta_{m,n}^{(\ell)}$, within the range $(-1,1]$, is the cosine of the AoD \cite{balanis2016antenna}. $L_{m,n}$ is the total number of the MPCs. $\mathbf{a}(\cdot)$ is the steering vector functions defined as
\begin{align} \label{eq_steeringVCT}
&\mathbf{a}(\theta)=[e^{j2\pi0(d/\lambda)\theta},e^{j2\pi(d/\lambda)\theta},\cdots,e^{j2\pi(N-1)(d/\lambda)\theta}]^{\mathrm{T}},
\end{align}
which depends on the array geometry. $d$ is the antenna spacing, and $\lambda$ is the signal wavelength. For a half-wavelength antenna spacing array, we have $d=\lambda/2$.

\subsection{Achievable Rate}
In general, the optimal decoding order for NOMA is the increasing order of the users' channel gains\cite{saito2013non,Dai2018NOMAsurvey}. However, for the mmWave-NOMA with HBF structure in this paper, the effective channel gains of the users are determined by both the channel gains and the beamforming gains. Thus, we need to sort the effective channel gains first, and then determine the decoding order. For notational simplicity and without loss of generality, we assume that the order of the effective channel gains in the $m$th group is $|\mathbf{h}^{\mathrm{H}}_{m,1}\mathbf{w}_{m}|^{2} \geq |\mathbf{h}^{\mathrm{H}}_{m,2}\mathbf{w}_{m}|^{2} \geq \cdots \geq |\mathbf{h}^{\mathrm{H}}_{m,|\mathcal{G}_{m}|}\mathbf{w}_{m}|^{2}$ \footnote{We can always define the user with the $n$th highest effective channel gain in the $m$th group as the $n$th user in this group. Thus, this simplified subscript has no influence on the solution in this paper.}, and thus the optimal decoding order is the increasing order of the effective channel gains \cite{xiao2018mmWaveNOMA,Ding2017random,Daill2017}. Therefore, the $n$th user in the $m$th group can decode $s_{m,j}~(n+1 \leq j \leq |\mathcal{G}_{m}|)$ and then remove them from the received signal in a successive manner. The other signals are treated as interference. Thus, the signal to interference plus noise power ratio (SINR) of the $n$th user in the $m$th group can be written as
\begin{equation}\label{eq_SINR}
\gamma_{m,n}=\frac{|\mathbf{h}^{\mathrm{H}}_{m,n}\mathbf{w}_{m}|^{2}p_{m,n}}{|\mathbf{h}^{\mathrm{H}}_{m,n}\mathbf{w}_{m}|^{2}\sum \limits_{j=1}^{n-1}p_{m,j}+\sum \limits_{i\neq m}\sum \limits_{k=1}^{|\mathcal{G}_{i}|}|\mathbf{h}^{\mathrm{H}}_{m,n}\mathbf{w}_{i}|^{2}p_{i,k}+\sigma^{2}}.
\end{equation}

Note that Gaussian signalling is assumed for transmitting data here. As a result, the achievable rate of the $n$th user in the $m$th group is
\begin{equation}\label{eq_Rate}
R_{m,n}=\log_{2}(1+ \gamma_{m,n}).
\end{equation}

Finally, the ASR of the proposed mmWave-NOMA system is
\begin{equation}\label{eq_ASR}
R_{\rm{sum}}=\sum\limits_{m=1}^{M} \sum\limits_{n=1}^{|\mathcal{G}_m|} R_{m,n}.
\end{equation}

Note that in the proposed downlink mmWave-NOMA system, we assume that the CSI between the BS and the users is known by the BS, and thus user grouping, power allocation and beamforming can be accomplished at the BS. The channel-gain information and beamforming-gain information of the other users are not required at the user side. However, compared with the conventional OMA system, information about the decoding order and codebook of the prior users in the same group should be transmitted to each user to accomplish SIC, which results in extra overhead. The amount of overhead depends on the number of users with in the same NOMA group. In the proposed solution of this paper, a great number of users are divided into many NOMA groups, and the number of users within the same NOMA group is usually not large so as to maintain the performance. Hence, the extra overhead is in fact not high, especially in slow varying channel, where the decoding order and codebook are also slow varying, the overhead can be further reduced.

\section{User Grouping and Problem Formulation}
As the number of the users is larger than that of the RF chains, i.e., $K>M$, we need to schedule the user into $M$ groups. To this end, we propose an intuitive algorithm for user grouping first, and then formulate a problem to jointly optimize HBF and power allocation.

\subsection{User Grouping}
Due to the spacial directivity of the mmWave channel, the users whose channels are highly correlated should be assigned to the same group to make full use of the multiplexing gain, while the users whose channels are uncorrelated should be assigned to different groups to decrease the interference. The normalized channel correlation between User $i$ and User $j$ is defined as
\begin{equation}
C_{i,j}=\frac{\mathbf{h}^{\mathrm{H}}_{i}\mathbf{h}_{j}}{\|\mathbf{h}_{i}\|\|\mathbf{h}_{j}\|}.
\end{equation}

We use the K-means clustering algorithm to implement the user grouping, where the normalized channel correlation is defined as the measure \cite{kanungo2002efficient}. First, we select $M$ users randomly, denoted by $\{\Omega_{1},\Omega_{2},\cdots,\Omega_{M}\}$, as the representatives of the $M$ clusters. Then, the other users can be assigned to the cluster according to the normalized channel correlation. For instance, User $k$ should be assigned to the $m^{\star}$th cluster, where
\begin{equation}\label{cluster}
m^{\star}=\mathop{\mathrm{arg~max}}\limits_{1\leq m \leq M} C_{k,\Omega_{m}}.
\end{equation}

After that, the representative of each cluster should be updated. To further decrease the correlation of the channels between different clusters, the representative of each cluster is updated as the one with the lowest correlation with the other clusters. The correlation between a user to the other clusters is defined as the summation of the normalized channel correlation between this user to the users of the other clusters, i.e.,
\begin{equation}
\bar{C}_{k}=\sum \limits^{j \notin \mathcal{G}^{(k)} }_{1 \leq j \leq K}C_{k,j},
\end{equation}
where $\mathcal{G}^{(k)}$ denotes the cluster which includes User $k$, and the representative of the $m$th cluster is updated as
\begin{equation}\label{representative}
\Omega_{m}=\mathop{\mathrm{arg~min}}\limits_{1\leq n \leq |\mathcal{G}_{m}|} \bar{C}_{n},
\end{equation}
where $\mathcal{G}_{m}$ denotes the $m$th cluster. After updating the representative of each cluster, the other users are reassigned to the clusters according to \eqref{cluster}. The iteration is stopped if the representatives of the clusters are unchanged. The details of the proposed user grouping algorithm are summarized in Algorithm \ref{alg_grouping}.

\begin{algorithm}[h]
\caption{User Grouping Algorithm}
\label{alg_grouping}
\begin{algorithmic}[1]
\REQUIRE ~\\
$K$, $M$, $\{\mathbf{h}_{k}\}$, and $\{C_{i,j}\}$.
\ENSURE ~The user grouping scheme: $\{\mathcal{G}_1,\mathcal{G}_2,\cdots,\mathcal{G}_M\}$.\\
\STATE $\mathcal{K}=\{1,2,\cdots,K\}$.
\STATE Initialize $\Omega_{m}^{(1)}=k_{m} \in \mathcal{K}$ randomly for $m=1,2,\cdots,M$.
\STATE $t=1$.
\WHILE {$\{\Omega_{m}^{(t)}\}\neq \{\Omega_{m}^{(t-1)}\}$}
\STATE Initialize $\mathcal{G}_{m}=\Omega_{m}^{(t)}$ for $m=1,2,\cdots,M$.
\FOR {$k \in \mathcal{K}/\{\Omega_{m}^{(t)}$\} }
\STATE $m^{\star}=\mathop{\mathrm{arg~max}}\limits_{1\leq m \leq M} C_{k,\Omega_{m}^{(t)}}$.
\STATE $\mathcal{G}_{m^{\star}}=\mathcal{G}_{m^{\star}}\bigcup k$.
\ENDFOR
\STATE $t=t+1$.
\STATE Update $\Omega_{m}^{(t)}$ for $m=1,2,\cdots,M$ according to \eqref{representative}.
\ENDWHILE
\RETURN $\{\mathcal{G}_1,\mathcal{G}_2,\cdots,\mathcal{G}_M\}$.
\end{algorithmic}
\end{algorithm}

\subsection{Problem Formulation}
Generally, there are mainly two categories of optimizing the overall rate performance in a communication system. One is to maximize the ASR. However, when maximizing the sum rate, the BS tends to allocate most power and beam gains to the users with the strong channels. Then, the users with the low channel gains can not be served by the BS. The other category is to ensure the user fairness, where the max-min fairness or proportion fairness are considered to improve the performance of the users with worse channel conditions. However, the fairness among the users may result in a performance loss of the sum rate. To realize the tradeoff between the sum-rate performance and the user fairness, we maximize the achievable sum rate while ensuring the minimum achievable rate of each user in this paper, which is also adopted in the related mmWave-NOMA systems \cite{Daill2017,Dai2018MIMONOMA,Cui2018mmWaveNOMA}. Then, the problem is formulated as
\begin{equation}\label{eq_problem}
\begin{aligned}
\mathop{\mathrm{Max}}\limits_{\{p_{m,n}\},\mathbf{A},\mathbf{D}}~~~~ &R_{\rm{sum}}\\
\mathrm{s.t.}~~~~~~~~ &C_1~:~R_{m,n} \geq r_{m,n},~~\forall m,n, \\
&C_2~:~p_{m,n} \geq 0, ~~\forall m,n,\\
&C_3~:~\sum\limits_{m=1}^{M} \sum\limits_{n=1}^{|\mathcal{G}_m|} p_{m,n} \leq P, \\
&C_4~:~|[\mathbf{A}]_{i,j}| = \frac{1}{\sqrt{N}},~~\forall i,j,\\
&C_5~:~\|[\mathbf{AD}]_{:,m}\| = 1,~~\forall m,
\end{aligned}
\end{equation}
where the constraint $C_1$ is the minimum rate constraint for each user. The constraint $C_2$ indicates that the power allocated to each user should be non-negative. The constraint $C_3$ is the total transmission power constraint, where the total power at the BS is no more than $P$. $C_4$ is the CM constraint for the ABF matrix, and $C_5$ is the unit power constraint for the HBF matrix.

The total dimension of the variables in Problem \eqref{eq_problem} is $K+MN+M^2$, which is large in general. Exhaustive search for the optimal solution results in heavy computational load, which is hard to accomplish in practice. To solve Problem \eqref{eq_problem}, there are two main challenges. One is that the optimized variables are entangled with each other, which makes the formulation non-convex. The other is that the expression of $R_{\rm{sum}}$ depends on the decoding order. In general, the optimal decoding order is the increasing order of the users' effective channel gains. However, the order of effective channel gains varies with different beamforming matrixes. In other words, given different HBF matrixes, the objective function in Problem \eqref{eq_problem}, i.e., the ASR of the users, has different expressions. The two challenges make it infeasible to solve Problem \eqref{eq_problem} by using the existing optimization tools. Next, we will propose a sub-optimal solution with promising performance but low computational complexity.

The proposed solution of Problem \eqref{eq_problem} can be obtained with two stages. In the first stage, we provide a low-complexity algorithm to obtain the sub-optimal power allocation with an arbitrary fixed HBF. In the second stage, we design the HBF, where the DBF matrix and the ABF matrix are obtained using the AZF method and the proposed BC-PSO algorithm, respectively.

\section{Solution of Power Allocation}
As we have analyzed before, an essential challenge to solve Problem \eqref{eq_problem} is the variation of the decoding order. However, given an arbitrary fixed ABF matrix ${\bf{A}}$ and an arbitrary fixed DBF matrix ${\bf{D}}$, the order of the effective channel gains is fixed. For notational simplicity and without loss of generality, we assume $|\mathbf{h}^{\mathrm{H}}_{m,1}\mathbf{w}_{m}|^{2} \geq |\mathbf{h}^{\mathrm{H}}_{m,2}\mathbf{w}_{m}|^{2} \geq \cdots \geq |\mathbf{h}^{\mathrm{H}}_{m,|\mathcal{G}_{m}|}\mathbf{w}_{m}|^{2}$ for any $1\leq m\leq M$, where $\mathbf{w}_{m}=[\mathbf{AD}]_{:,m}$. The original problem can be simplified as
\begin{equation}\label{eq_problem2}
\begin{aligned}
\mathop{\mathrm{Max}}\limits_{\{p_{m,n}\}}~~~~ &R_{\rm{sum}}\\
\mathrm{s.t.}~~~~~~ &C_1~:~R_{m,n} \geq r_{m,n},~~\forall m,n, \\
&C_2~:~p_{m,n} \geq 0, ~~\forall m,n,\\
&C_3~:~\sum\limits_{m=1}^{M} \sum\limits_{n=1}^{|\mathcal{G}_m|} p_{m,n} \leq P,
\end{aligned}
\end{equation}
where ${\bf{A}}$ and ${\bf{D}}$ are arbitrary and fixed.

According to the expression of the achievable rate in \eqref{eq_Rate}, a user may suffer the interference from both the intra-group users and the inter-group users. Although the HBF matrix is fixed, the objective function and the constraint $C_1$ of Problem \eqref{eq_problem2} are still non-convex. To address this problem, we divide it into two sub-problems, i.e., intra-GPA and inter-GPA. Define $\sum\limits_{n=1}^{|\mathcal{G}_m|} p_{m,n}=P_{m}$ for $1\leq m\leq M$, which means the allocated power for the $m$th group, and then Problem \eqref{eq_problem2} is equivalent to
\begin{equation}\label{eq_problem3}
\begin{aligned}
\mathop{\mathrm{Max}}\limits_{\{P_{m}\}}\mathop{\mathrm{Max}}\limits_{\{p_{m,n}\}}~~~~ &R_{\rm{sum}}\\
\mathrm{s.t.~~~~~~~}~~~ &C_1~:~R_{m,n} \geq r_{m,n},~~\forall m,n, \\
&C_2~:~p_{m,n} \geq 0, ~~\forall m,n,\\
&C_3~:~\sum\limits_{n=1}^{|\mathcal{G}_m|} p_{m,n} = P_{m}, ~~\forall m,\\
&C_4~:~\sum\limits_{m=1}^{M} P_{m} \leq P,
\end{aligned}
\end{equation}

Note that the introduced inter-GPA variables, i.e., $\{P_{m}\}$, have no influence on the optimality of the power allocation problem, because there is no loss of the degree of freedom in Problem \eqref{eq_problem3} compared with Problem \eqref{eq_problem2}, and Problem \eqref{eq_problem3} is more tractable. First, given arbitrary and fixed inter-GPA, a closed-form sub-optimal intra-GPA can be obtained. Then, substituting the intra-GPA into Problem \eqref{eq_problem3}, we can obtain a sub-optimal inter-GPA solution. Although the proposed solution of power allocation is not globally optimal, we will prove that it is near-to-optimal when the inter-group interference is small through the theoretical analysis and simulation verification.

\subsection{The Intra-GPA Problem}
As shown in \eqref{eq_SINR} and \eqref{eq_Rate}, one user may suffer the interference from the users in the same group and the users in other groups, which are called intra-group interference and inter-group interference, respectively. Considering that HBF can be well designed in general, such that the inter-group interference is small and can be neglected. Thus, we have the following proposition to solve the intra-GPA problem.

\begin{proposition} Given an arbitrary fixed inter-GPA of $\{P_1,P_2,\cdots,P_M\}$, if the inter-group interference can be neglected, the optimal intra-GPA in Problem \eqref{eq_problem3} should always satisfy
\begin{equation}
R_{m,n}=r_{m,n}~(1\leq m \leq M,~2\leq n \leq |\mathcal{G}_m|).
\end{equation}
\end{proposition}
\begin{proof}
If the inter-group interference is small and can be neglected, Problem \eqref{eq_problem3} can be divided into $M$ independent intra-GPA problems. For the $m$th group, the intra-GPA problem is simplified as
\begin{equation}\label{eq_problem4}
\begin{aligned}
\mathop{\mathrm{Max}}\limits_{\{p_{m,n}\}}~~~~ &\sum \limits _{n=1}^{|\mathcal{G}_{m}|}R_{m,n}\\
\mathrm{s.t.~~~}~~~ &C_1~:~R_{m,n} \geq r_{m,n},~~\forall n, \\
&C_2~:~p_{m,n} \geq 0, ~~\forall n,\\
&C_3~:~\sum\limits_{n=1}^{|\mathcal{G}_{m}|} p_{m,n} = P_{m},
\end{aligned}
\end{equation}
which is a power allocation problem without inter-group interference. This problem has been solved in \cite{Zhu2018NOMAPSO}, where the optimal power allocation always satisfies $R_{m,n}=r_{m,n} ~(2\leq n \leq |\mathcal{G}_m|)$.
\end{proof}

By solving the equation sets of $R_{m,n}=r_{m,n}~(1\leq m \leq M,~2\leq n \leq |\mathcal{G}_m|)$ and $\sum\limits_{n=1}^{|\mathcal{G}_{m}|} p_{m,n} = P_{m}~(1\leq m \leq M)$, we can obtain a sub-optimal intra-GPA for each group of users, which is shown in \eqref{opt_power} on the top of the next page, where $\eta_{m,n}=2^{r_{m,n}}-1$. Note that although the inter-group interference is neglected in Proposition 1, it is included when solving the equation sets. Thus, the minimal rate constraints for the users (from the 2nd one to the last one in each group) are always satisfied. The impact of the approximation on the inter-group interference will be evaluated in the simulation.

\begin{figure*}
\begin{equation}\label{opt_power}
\left \{
\begin{aligned}
&p_{m,|\mathcal{G}_{m}|}^{\circ}=\frac{\eta_{m,|\mathcal{G}_{m}|}}{\eta_{m,|\mathcal{G}_{m}|}+1}(P_{m}+\frac{\sum \limits_{i\neq m}|\mathbf{h}^{\mathrm{H}}_{m,|\mathcal{G}_{m}|}\mathbf{w}_{i}|^{2}P_{i}+\sigma^2}{|\mathbf{h}^{\mathrm{H}}_{m,|\mathcal{G}_{m}|}\mathbf{w}_{m}|^{2}}),\\
&p_{m,|\mathcal{G}_{m}|-1}^{\circ}=\frac{\eta_{m,|\mathcal{G}_{m}|-1}}{\eta_{m,|\mathcal{G}_{m}|-1}+1}(P_{m}-p_{m,|\mathcal{G}_{m}|}^{\circ}+\frac{\sum \limits_{i\neq m}|\mathbf{h}^{\mathrm{H}}_{m,|\mathcal{G}_{m}|-1}\mathbf{w}_{i}|^{2}P_{i}+\sigma^2}{|\mathbf{h}^{\mathrm{H}}_{m,|\mathcal{G}_{m}|-1}\mathbf{w}_{m}|^{2}}),\\
&~~~~\vdots\\
&p_{m,2}^{\circ}=\frac{\eta_{m,2}}{\eta_{m,2}+1}(P_{m}-\sum\limits_{k=3}^{|\mathcal{G}_{m}|}p_{m,k}^{\circ}+\frac{\sum \limits_{i\neq m}|\mathbf{h}^{\mathrm{H}}_{m,2}\mathbf{w}_{i}|^{2}P_{i}+\sigma^2}{|\mathbf{h}^{\mathrm{H}}_{m,2}\mathbf{w}_{m}|^{2}}),\\
&p_{m,1}^{\circ}=P_{m}-\sum\limits_{k=2}^{|\mathcal{G}_{m}|}p_{m,k}^{\circ},
\end{aligned}
\right.
\end{equation}
\hrulefill
\end{figure*}

Under Proposition 1, the ASR in Problem \eqref{eq_problem3} can be simplified as
\begin{equation}\label{eq_ASR2}
R_{\rm{sum}}=\sum\limits_{m=1}^{M}R_{m,1}+\sum\limits_{m=1}^{M} \sum\limits_{n=2}^{|\mathcal{G}_m|}r_{m,n}.
\end{equation}

Substituting \eqref{opt_power} into Problem \eqref{eq_problem3}, Problem \eqref{eq_problem3} can be transformed to
\begin{equation}\label{eq_problem5}
\begin{aligned}
\mathop{\mathrm{Max}}\limits_{\{P_{m}\}}~~~~ &\sum\limits_{m=1}^{M}R_{m,1}\\
\mathrm{s.t.~~~}~~~ &C_1~:~R_{m,1} \geq r_{m,1},~~\forall m, \\
&C_2~:~\sum\limits_{m=1}^{M} P_{m} \leq P,
\end{aligned}
\end{equation}
which is an inter-GPA problem.

\subsection{The Inter-GPA Problem}
Due to the inter-group interference in the expression of the objective function, it is still challenging to solve Problem \eqref{eq_problem5}. We propose an iterative algorithm here. First, we initialize the group power $P_{m}$ equally. Then, we start iteration. In each iteration, the inter-group interference is assumed to be invariable, and we update the inter-GPA by maximizing the ASR in Problem \eqref{eq_problem5}, where the inter-group interference is defined as
\begin{equation}
\begin{aligned}
I^{\mathrm{(inter)}}_{m,n}\triangleq \sum \limits_{i\neq m}\sum \limits_{k=1}^{|\mathcal{G}_{i}|}|\mathbf{h}^{\mathrm{H}}_{m,n}\mathbf{w}_{i}|^{2}p_{i,k}=\sum \limits_{i\neq m}|\mathbf{h}^{\mathrm{H}}_{m,n}\mathbf{w}_{i}|^{2}P_{i}.
\end{aligned}
\end{equation}

Thus, the SINR for the first user in each group is linear to its signal power, i.e.,
\begin{equation}
\gamma_{m,1}=\frac{|\mathbf{h}^{\mathrm{H}}_{m,1}\mathbf{w}_{m}|^{2}p_{m,1}^{\circ}}{I^{\mathrm{(inter)}}_{m,1}+\sigma^2}
\end{equation}
where $p_{m,1}^{\circ}$ is defined in \eqref{opt_power}. Furthermore, according to the expression in \eqref{opt_power}, if the inter-group interference is invariable,  $p_{m,1}^{\circ}$ is also linear to $P_{m}$. Thus, we can obtain the relationship between $\gamma_{m,1}$ and $P_{m}$ as
\begin{equation}\label{SINR_linear2}
\gamma_{m,1}=k_{m}P_{m}+b_{m},
\end{equation}
where $k_{m}$ and $b_{m}$ are given by
\begin{align}
&k_{m}=\frac{|\mathbf{h}^{\mathrm{H}}_{m,1}\mathbf{w}_{m}|^{2}}{I^{\mathrm{(inter)}}_{m,1}+\sigma^2}\Bigg{(}1-\sum \limits_{n=2}^{|\mathcal{G}_{m}|} \Big{[} \eta_{m,n}\prod \limits_{j=2}^{n}\frac{1}{(\eta_{m,j}+1)}\Big{]} \Bigg{)},\\ \nonumber
&b_{m}=-\frac{|\mathbf{h}^{\mathrm{H}}_{m,1}\mathbf{w}_{m}|^{2}}{I^{\mathrm{(inter)}}_{m,1}+\sigma^2} \times \sum \limits_{n=2}^{|\mathcal{G}_{m}|} \Big{[} \eta_{m,n}\frac{I^{\mathrm{(inter)}}_{m,n}+\sigma^2}{|\mathbf{h}^{\mathrm{H}}_{m,n}\mathbf{w}_{n}|^{2}}\prod \limits_{j=2}^{n}\frac{1}{(\eta_{m,j}+1)}\Big{]}.
\end{align}

It is easy to verify that $k_{m}>0$ and $b_{m}<0$. Then, the objective function in Problem \eqref{eq_problem5} is equal to
\begin{equation}
\begin{aligned}
\sum\limits_{m=1}^{M}R_{m,1}&=\sum\limits_{m=1}^{M}\log_{2}(1+\gamma_{m,1}) =\sum\limits_{m=1}^{M}\log_{2}(k_{m}P_{m}+b_{m}+1)\triangleq f(\{P_{m}\}).
\end{aligned}
\end{equation}

Constraint $C_1$ in Problem \eqref{eq_problem5} is equivalent to
\begin{equation}
\begin{aligned}
&R_{m,1} \geq r_{m,1}
\Leftrightarrow \gamma_{m,1} \geq \eta_{m,1}
\Leftrightarrow P_{m} \geq \frac{\eta_{m,1}-b_{m}}{k_{m}}.
\end{aligned}
\end{equation}

As the objective function becomes concave now and the constraints are linear, Problem \eqref{eq_problem5} can be directly solved by using the convex optimization tools \cite{boyd2004convex}. In order to explore the essential principle of the inter-GPA for mmWave-NOMA, we propose a method with low computation complexity here. We begin from the case without constraint $C_1$ in Problem \eqref{eq_problem5} and give the following Lemma.
\begin{lemma}
If the inter-group interference is assumed to be invariant in Problem \eqref{eq_problem5}, without the constraint $C_1$, the globally optimal solution is
\begin{equation}\label{inter_power}
P_{m}^{\star}=\frac{P+\sum \limits_{i=1}^{M}\frac{b_{i}+1}{k_{i}}}{M}-\frac{b_{m}+1}{k_{m}}, ~1\leq m \leq M.
\end{equation}
\end{lemma}

\begin{proof}
See Appendix A.
\end{proof}

According to Lemma 1, if $P_{m}^{\star}$ in \eqref{inter_power} is located in the feasible domain of the constraint $C_1$ in Problem \eqref{eq_problem5}, i.e., $P_{m}^{\star}\geq \frac{\eta_{m,1}-b_{m}}{k_{m}}$ for all $1\leq m\leq M$, $P_{m}^{\star}$ is the optimal solution of Problem \eqref{eq_problem5}. However, if $P_{m}^{\star}$ in \eqref{inter_power} is not located in the feasible domain of the constraint $C_1$ in Problem \eqref{eq_problem5}, i.e. $P_{m}^{\star}>\frac{\eta_{m,1}-b_{m}}{k_{m}}$ for any one of $1\leq m\leq M$, $P_{m}^{\star}$ is not the optimal solution of Problem \eqref{eq_problem5}. We may find the optimal solution by using the following Lemma.

\begin{lemma}
If the inter-group interference is assumed to be invariant in Problem \eqref{eq_problem5}, with the constraint $C_1$, the globally optimal solution should always satisfy
\begin{equation}\label{inter_power2}
P_{m}^{\circ}=\frac{\eta_{m,1}-b_{m}}{k_{m}}, ~\forall m \in \mathcal{U},
\end{equation}
where $\mathcal{U}=\{i|1\leq i\leq M,~P_{i}^{\star}<\frac{\eta_{i,1}-b_{i}}{k_{i}}\}$ and $P_{i}^{\star}$ is defined in \eqref{inter_power}.
\end{lemma}

\begin{proof}
See Appendix B.
\end{proof}

Lemma 2 provides the globally optimal power allocation for $m \in \mathcal{U}$. For $m \notin \mathcal{U}$, the optimal power allocation can be obtained by solving the following problem.
\begin{equation}\label{eq_problem5.2}
\begin{aligned}
\mathop{\mathrm{Max}}\limits_{\{P_{m}\}}~~~~ &\sum\limits_{m \notin \mathcal{U}} R_{m,1}\\
\mathrm{s.t.~~~}~~~ &C_1~:~R_{m,1} \geq r_{m,1},~~m \notin \mathcal{U}, \\
&C_2~:~\sum\limits_{m \notin \mathcal{U}} P_{m} \leq P-\sum\limits_{j \in \mathcal{U}} P_{j}^{\circ},
\end{aligned}
\end{equation}
which has a similar formulation with Problem \eqref{eq_problem5}. Thus, Lemma 1 and Lemma 2 can also be used to solve Problem \eqref{eq_problem5.2}, which forms a closed loop. In summary, we give Algorithm \ref{alg_powerallo} to accomplish the inter-GPA.

\begin{algorithm}[h]
\caption{Inter-GPA}
\label{alg_powerallo}
\begin{algorithmic}[1]
\REQUIRE ~$K$, $M$, $\{\mathcal{G}_m\}$, $P$, $\{\mathbf{h}_{k}\}$, $\{r_{k}\}$, $\mathbf{W}$, and $F_{\mathrm{max}}$.\\
\ENSURE ~Inter-GPA: $\{P_{m}^{\circ}\}$.\\
\STATE $P_{m}^{\circ(0)}=\frac{P}{M} ~(1\leq m \leq M)$.
\FOR {$t=1:F_{\mathrm{max}}$}
\STATE $\mathcal{M}=\{1,2,\cdots,M\}$.
\STATE $\mathcal{U}=\mathcal{M}$.
\WHILE {$\mathcal{U}\neq \Phi$}
\STATE Obtain $k_{m},~b_{m} ~(\forall m \in \mathcal{M})$ in \eqref{SINR_linear2}.
\STATE Obtain $P_{m}^{\star} ~(\forall m \in \mathcal{M})$ according to \eqref{inter_power}.
\STATE $\mathcal{U}=\{i|i \in \mathcal{M},~P_{i}^{\star}<\frac{\eta_{i,1}-b_{i}}{k_{i}}\}$.
\STATE $P_{m}^{\circ(t)}=\frac{\eta_{m,1}-b_{m}}{k_{m}} ~(\forall m \in \mathcal{U})$.
\STATE $\mathcal{M}=\mathcal{M}/\mathcal{U}$.
\ENDWHILE
\STATE $P_{m}^{\circ(t)}=P_{m}^{\star} ~(\forall m \in \mathcal{M})$.
\ENDFOR
\STATE $P_{m}^{\circ}=P_{m}^{\circ(T_{\mathrm{max}})} ~(1 \leq m \leq M)$.
\RETURN $\{P_{m}^{\circ}\}$.
\end{algorithmic}
\end{algorithm}

Hereto, the power allocation is solved. Given an arbitrary fixed HBF, we can obtain the inter-GPA using Algorithm \ref{alg_powerallo} and obtain the intra-GPA according to \eqref{opt_power}. Since the proposed intra-GPA and inter-GPA solutions are both sub-optimal, we provide the following theorem to evaluate the optimality of the proposed power allocation solution.
\begin{theorem}
If the inter-group interference in Problem \eqref{eq_problem3} is small and approaches to zero, the proposed solution of power allocation in Algorithm \ref{alg_powerallo} and \eqref{opt_power} is globally optimal.
\end{theorem}

\begin{proof}
If the inter-group interference is zero, the intra-GPA problems are independent for different groups. According to the conclusion in [21, Theorem 1], \eqref{opt_power} is the optimal intra-GPA solution with the given fixed inter-GPA. Substituting \eqref{opt_power} into Problem \eqref{eq_problem3}, the inter-GPA problem is concave and can be solved by using Algorithm \ref{alg_powerallo} with only one iteration. Due to the concavity, the inter-GPA solution is also optimal. Thus, the globally optimal power allocation can be obtained by using the proposed scheme if the inter-group interference is zero.
\end{proof}

Based on Theorem 1, we can find that the optimality of the power allocation solution depends on the inter-group interference, which can be restrained through the elaborate beamforming design. Thus, the design of HBF should take both decreasing the interference and increasing the ASR into account. The details will be shown in the next section.

\section{Solution of Hybrid Beamforming}
In this Section, we provide the solution of HBF in Problem \eqref{eq_problem}. As we have analyzed previously, the design of HBF should guarantee the suppression of the inter-group interference, as well as the improvement of the ASR. For mmWave-NOMA, there may exist more than one users in each group. The traditional unidirectional beamforming cannot support all the users. Thus, a multi-directional beamforming scheme is required in the analog domain. However, the non-convex modulus constraint for ABF makes the beamforming problem challenging. Besides, as shown in \eqref{eq_SINR}, due to the superposition of the inter-group interference and the intra-group interference, it is difficult to obtain the optimal HBF solution. To this end, we propose a sub-optimal approach. First, the DBF is designed using the AZF method to reduce the inter-group interference, where the ABF matrix is arbitrary and fixed. Then, we use the BC-PSO algorithm in \cite{Zhu2018NOMAPSO} to solve the ABF problem, where the power allocation and DBF matrix are substituted as the function of the ABF matrix.

\subsection{DBF with Arbitrary Fixed ABF}
As each group of users have a unique DBF vector, we may design the DBF with the AZF method to reduce the inter-group interference, where the ABF is arbitrary and fixed. Since the rank of the DBF matrix is no more than the number of the users, i.e., $M\leq K$, the inter-group interference cannot be completely suppressed through DBF. Recalling that when optimizing the power allocation, the rate gains are acquired at the first user in each group. Thus, we select the channel response vector of the user with the highest channel gain in each group as the equivalent channel vector. Note that the channel gain utilized here corresponds to the power of the channel response vector before beamforming, which differs from the effective channel gain after beamforming. Then, the $N\times M$ equivalent channel matrix is
\begin{equation}
\mathbf{\tilde{H}}=[\mathbf{h}_{1,1},\mathbf{h}_{2,1},\cdots,\mathbf{h}_{M,1}].
\end{equation}

Consequently, the DBF matrix can be generated by the AZF method as \footnote{Since the DBF design implements an approximate zero-forcing method, i.e., only to the first user in each group, the inclusion of inter-group interference in the previous section is relevant.}
\begin{equation}\label{DBF}
\mathbf{\tilde{D}}=(\mathbf{\tilde{H}}^{\mathrm{H}}\mathbf{A})^{\dag}.
\end{equation}

Due to the unit power constraint for the HBF matrix, each column of the DBF matrix should be normalized as
\begin{equation}\label{DBF_norm}
[\mathbf{D}^{\circ}]_{:,m}=\frac{[\mathbf{\tilde{D}}]_{:,m}}{\|\mathbf{A}[\mathbf{\tilde{D}}]_{:,m}\|}.
\end{equation}

Although the inter-group interference cannot be completely eliminated with DBF, it can be further suppressed with ABF, which has a higher degree of freedom.

\subsection{ABF Using BC-PSO Alogrithm}
Given an arbitrary fixed ABF matrix, we can obtain the DBF matrix according to \eqref{DBF} and \eqref{DBF_norm}. Then, the inter-GPA can be obtained by Algorithm \ref{alg_powerallo}, and meanwhile the intra-GPA is given by \eqref{opt_power}. It is hard to optimize ABF with the conventional approaches, since the closed-form expression of $R_{\mathrm{sum}}$ over $\mathbf{A}$ is complicated. In addition, the ABF matrix $\mathbf{A}$ with CM constraint is high-dimensional, i.e., $N \times M$, which makes the ABF design difficult.

To solve this difficult problem, particle swarm optimization (PSO) is a good approach \cite{fukuyama2008fundamentals}. In the $N\times M$-dimensional search space $\mathcal{S}$, the $I$ particles in the swarm are randomly initialized with position $\mathbf{A}$ and velocity $\mathbf{V}$. Each particle has a memory for its best found position $\mathbf{P}_{\text{best}}$ and the globally best position $\mathbf{G}_{\text{best}}$, where the goodness of a position is evaluated by the fitness function. For each iteration, the velocity and position of each particle are updated based on
\begin{equation}\label{eq_PSOregular}
\begin{aligned}
&[\mathbf{V}]_{i,j}=\omega[\mathbf{V}]_{i,j}+c_{1}\text{rand()}*([\mathbf{P}_{\text{best}}]_{i,j}-[\mathbf{A}]_{i,j})+c_{2}\text{rand()}*([\mathbf{G}_{\text{best}}]_{i,j}-[\mathbf{A}]_{i,j})\\
&[\mathbf{A}]_{i,j}=[\mathbf{A}]_{i,j}+[\mathbf{V}]_{i,j}
\end{aligned}
\end{equation}
for $i=1,2,\cdots,N;~j=1,2,\cdots,M$. The parameter $\omega$ is the inertia weight of velocity. In general, $\omega$ is decreasing linearly from the maxima to the minima for each time of iteration to improve the convergence speed. The parameters $c_{1}$ and $c_{2}$ are the cognitive ratio and social ratio, respectively. The random number function rand() returns a number between 0.0 and 1.0 with uniform distribution.

Due to the CM constraint, the search space for $\mathbf{A}$, i.e., $\{\mathbf{A}\big{|}|[\mathbf{A}]_{i,j}| = \frac{1}{\sqrt{N}}\}$, is highly non-convex. It has been shown that the BC-PSO algorithm outperforms the classic PSO algorithm in the ABF problem \cite{Zhu2018NOMAPSO}. The key idea of the BC-PSO algorithm is to relax the search space as a convex set, i.e., $\mathcal{S}=\{\mathbf{A}\big{|}|[\mathbf{A}]_{i,j}| \leq \frac{1}{\sqrt{N}}\}$, and adjust the particles onto the boundaries for each iteration to satisfy the CM constraint. The outer boundary is defined as $\{\mathbf{A}\big{|}|[\mathbf{A}]_{i,j}| = d_{\mathrm{out}}\}$, where $d_{\mathrm{out}}=\frac{1}{\sqrt{N}}$ is fixed. The inter boundary is defined as $\{\mathbf{A}\big{|}|[\mathbf{A}]_{i,j}| = d_{\mathrm{in}}\}$, where $d_{\mathrm{in}}=\frac{t}{T_{\mathrm{max}}}\frac{1}{\sqrt{N}}$ is dynamic. $T_{\mathrm{max}}$ is the maximum number of iterations and $t=1,2,\cdots,T_{\mathrm{max}}$. For each iteration, the particles out of the boundaries are adjusted onto the boundaries. Then, after calculating the fitness function for each particle, the locally and globally best positions, i.e., $\mathbf{P}_{\text{best}}$ and $\mathbf{G}_{\text{best}}$, are updated. With this implementation, the particles can move throughout the relaxed search space and converge to satisfy the CM constraint eventually. Compared with the classic PSO algorithm, the BC-PSO algorithm has enhanced search capabilities.

\section{Summary of the Complete Solution and Computational Complexity}
\subsection{Summary of The Complete Solution}
In the above sections, we have presented the algorithms and formulas, respectively, for user grouping, power allocation, digital beamforming and analog beamforming. Based on these algorithms and formulas, we give the complete solution to realize an arbitrary mmWave-NOMA system. As shown in in Algorithm \ref{alg_PSO}, we firstly use Algorithm \ref{alg_grouping} to divide the users into $M$ groups, and obtain $\{\mathcal{G}_m\}$. Then, we use the BC-PSO algorithm to iteratively optimize the position of the particle, i.e., the ABF matrix, where the fitness function is defined as the ASR in \eqref{eq_ASR}. Note that in the part of power allocation and DBF, we assume that the ABF matrix is arbitrary and fixed. Thus, the power allocation and DBF can be substituted as the function of the analog beamforming matrix in Algorithm \ref{alg_PSO}. Given different ABF matrixes, we should calculate the power allocation and DBF matrixes first, and then obtain the ASR. In each iteration, the computations of the DBF matrix $\mathbf{D}^{\circ}$ using \eqref{DBF} and \eqref{DBF_norm}, the inter-GPA $\{P_{m}^{\circ}\}$ using Algorithm \ref{alg_powerallo}, and the intra-GPA $\{p_{m,n}^{\circ}\}$ using \eqref{opt_power} are performed sequentially after determining the ABF matrix. Hence, after $T_{\mathrm{max}}$ iterations, the sub-optimal overall solution $\mathbf{A}^{\circ}$, $\mathbf{D}^{\circ}$ and $\{p_{m,n}^{\circ}\}$ are jointly obtained.

\begin{algorithm}[h]
\caption{Proposed solution for mmWave-NOMA}
\label{alg_PSO}
\begin{algorithmic}[1]
\REQUIRE ~$K$, $M$, $N$, $P$, $\{\mathbf{h}_{k}\}$, $\{r_{k}\}$, and parameters \\
~~~~~~      for BC-PSO $\{I$, $T_{\mathrm{max}}$, $c_1$, $c_2$, $\omega_{\text{max}}$, $\omega_{\text{min}}\}$.
\ENSURE ~$\{\mathcal{G}_m\}$, $\mathbf{A}^{\circ}$, $\mathbf{D}^{\circ}$ and $\{p_{m,n}^{\circ}\}$.\\
\STATE Obtain the user grouping $\{\mathcal{G}_m\}$ using Algorithm \ref{alg_grouping}.
\STATE Initialize the position $\mathbf{A}_{i}$ and velocity $\mathbf{V}_{i}$.
\STATE Find the globally best position $\mathbf{G}_{\text{best}}$.
\FOR {$t=1:T_{\mathrm{max}}$ }
\STATE $\omega=\omega_{\text{max}}-\frac{t}{T}(\omega_{\text{max}}-\omega_{\text{min}})$.
\STATE $d_{\mathrm{out}}=\frac{1}{\sqrt{N}},~d_{\mathrm{in}}=\frac{t}{T_{\mathrm{max}}}\frac{1}{\sqrt{N}}$.
\FOR {$l=1:I$ }
\FOR {$i=1:N$ }
\FOR {$j=1:M$ }
\STATE Update $[\mathbf{V}_{l}]_{i,j}$ and $[\mathbf{A}_{l}]_{i,j}$ based on \eqref{eq_PSOregular}.
\IF{$|[\mathbf{A}_{l}]_{i,j}|>d_{\mathrm{out}}$}
\STATE $[\mathbf{A}_{l}]_{i,j}=d_{\mathrm{out}}\frac{[\mathbf{A}_{l}]_{i,j}}{|[\mathbf{A}_{l}]_{i,j}|}$.
\ENDIF
\IF{$|[\mathbf{A}_{l}]_{i,j}|<d_{\mathrm{in}}$}
\STATE $[\mathbf{A}_{l}]_{i,j}=d_{\mathrm{in}}\frac{[\mathbf{A}_{l}]_{i,j}}{|[\mathbf{A}_{l}]_{i,j}|}$.
\ENDIF
\IF{$|[\mathbf{P}_{\text{best},l}]_{i,j}|<d_{\mathrm{in}}$}
\STATE $[\mathbf{P}_{\text{best},l}]_{i,j}=d_{\mathrm{in}}\frac{[\mathbf{P}_{\text{best},l}]_{i,j}}{|[\mathbf{P}_{\text{best},l}]_{i,j}|}$.
\ENDIF
\STATE Obtain the DBF matrix $\mathbf{D}^{\circ}$ according to \eqref{DBF} and \eqref{DBF_norm}.
\STATE Reorder the effective channel gains of the users in each group.
\STATE Obtain the inter-GPA $\{P_{m}^{\circ}\}$ using Algorithm \ref{alg_powerallo}.
\STATE Obtain the intra-GPA $\{p_{m,n}^{\circ}\}$ according to \eqref{opt_power}.
\STATE Obtain the fitness function $R_{\mathrm{sum}}$ according to \eqref{eq_ASR}.
\ENDFOR
\ENDFOR
\STATE Update $\mathbf{P}_{\text{best},l}$.
\ENDFOR
\STATE Update $\mathbf{G}_{\text{best}}$.
\ENDFOR
\STATE $\mathbf{A}^{\circ}=\mathbf{G}_{\text{best}}$.
\RETURN $\{\mathcal{G}_m\}$, $\mathbf{A}^{\circ}$, $\mathbf{D}^{\circ}$ and $\{p_{m,n}^{\circ}\}$.
\end{algorithmic}
\end{algorithm}

\subsection{Computational Complexity}
When operating the user grouping in Algorithm \ref{alg_grouping}, the complexities of calculating the channel correlation and the norm channel vector are $\mathcal{O}(K^{2}N)$ and $\mathcal{O}(KN)$, respectively. In each iteration, the complexities of updating the cluster representative and the user grouping are $\mathcal{O}(K^{2})$ and $\mathcal{O}(KM)$, respectively. Since the number of antennas is much larger than that of the RF chains, i.e., $N\gg M$, the maximal complexity of Algorithm \ref{alg_grouping} is $\mathcal{O}(K^{2}N)$. In Algorithm \ref{alg_powerallo}, the complexity of calculating the effective channel gains of the users is $\mathcal{O}(MKN)$. For each time of updating the inter-GPA, the maximal number of iterations to update the inter-GPA from Step 5 to 11 is $M$, and the complexity of computing the inter-GPA in each subcycle is no more than $\mathcal{O}(K^{2})$. Thus, the complexity of Algorithm \ref{alg_powerallo} is $\mathcal{O}(MKN+F_{\mathrm{max}}MK^{2})$. In Algorithm \ref{alg_PSO}, the numbers of invoking Algorithm \ref{alg_grouping} and Algorithm \ref{alg_powerallo} are 1 and $T_{\mathrm{max}}IMN$, respectively. Consequently, the total computational complexity of the proposed user pairing algorithm, HBF and power allocation algorithm is $\mathcal{O}(T_{\mathrm{max}}IM^{2}KN^{2}+T_{\mathrm{max}}F_{\mathrm{max}}IM^{2}K^{2}N)$, which is a polynomial complexity. In comparison, The total computational complexity of the algorithm in \cite{Dai2018MIMONOMA} is $\mathcal{O}(MK^2+MN+TK^{4.5}\log_2(1/\varepsilon))$, where $T$ is the maximum iteration times and $\varepsilon$ is the solution accuracy. Since the number of the antennas is much larger than those of the users and the RF chains, i.e., $N\gg K, N\gg M$, the computational complexity in \cite{Dai2018MIMONOMA} is lower compared with our algorithm, because the HBF is not jointly optimized with the power allocation.

\section{Simulation Results}
In this section, we provide some simulation results to verify the performance of the proposed mmWave-NOMA scheme. We adopt the channel model shown in \eqref{eq_oriChannel}, where the users are uniformly distributed from 10m to 100m away from the BS, and the channel gain of the node 30m away from the BS has an average power of 0 dB to noise power. The number of MPCs for all the users are $L=4$. Both LOS and NLOS channel models are considered. For the LOS channel, the average power of the NLOS paths is 15 dB weaker than that of the LOS path. For the NLOS channel, the coefficient of each path has an average power of $1/\sqrt{L}$. The cosine of the AoD for each path of the users is generated by a uniformly distributed random variable ranging from -1 to 1. Each point of the figures are the average performance of 100 channel realizations. The corresponding parameter settings are $I=800, F_{\mathrm{max}}=6, T_{\mathrm{max}}=200, c_{1}=c_{2}=1.4, \omega_{\text{max}}=0.9, \omega_{\text{min}}=0.4$.

In the simulations, we consider the following six typical mmWave communication schemes: ``mmWave-NOMA Proposed'' is corresponding to the proposed joint approach, including user grouping, power allocation, and HBF. ``mmWave-NOMA Ideal'' is based on the proposed joint approach and with assumption of none inter-group interference, i.e., $I^{\mathrm{(inter)}}_{m,n}=0$. Besides, ``mmWave-NOMA [13]'' and ``fully digital MIMO'' are corresponding to the approach for mmWave-NOMA with fully connected HBF structure in \cite{Dai2018MIMONOMA} and the mmWave-fully-digital-MIMO structure with zero-forcing precoding, respectively. For fair comparison, the power splitting part in \cite{Dai2018MIMONOMA} is neglected in the simulations, which means that all the power is used for wireless information transmission. ``TDMA-ZF'' corresponds to the performance of mmWave time division multiple access (TDMA) system, where $M$ out of $K$ users are served in each time slot. Each user is served by an independent analog beamformer with steering vector, and ZF and water-filling method is adopted for digital beamforming. While for ``mmWave-FDMA'', the users are assigned into $M$ groups, and the users in the same group perform frequency division multiple access (FDMA) \cite{Dai2018MIMONOMA}. Then, the achievable rate of the mmWave-FDMA scheme for the $k$th user is
\begin{equation}\label{eq_OMA}
R_{k}^{\mathrm{FDMA}}=\frac{1}{\left|\mathcal{G}^{k}\right|}\log_{2}\left(1+\frac{|\mathbf{h}^{\mathrm{H}}_{k}\mathbf{w}_{k}|^{2}p_{k}}{\sum \limits_{j\notin \mathcal{G}^{k}}|\mathbf{h}^{\mathrm{H}}_{k}\mathbf{w}_{j}|^{2}p_{j}+\frac{\sigma^{2}}{\left|\mathcal{G}^{k}\right|}}\right),
\end{equation}
where $\mathcal{G}^{k}$ represents the group which the $k$th user belongs to. The beamforming vector $\mathbf{w}_{k}$ and the power allocation $\{p_{k}\}$ are generated by using the approach in \cite{Dai2018MIMONOMA}.

In addition, we also evaluate the performance of the EE, which is defined as the ratio between the ASR and total power consumption, i.e.,
\begin{equation}\label{EE}
EE=\frac{R_{\mathrm{sum}}}{P+N_{\mathrm{RF}}P_{\mathrm{RF}}+N_{\mathrm{PS}}P_{\mathrm{PS}}},
\end{equation}
where $R_{\mathrm{sum}}$ is the ASR. $P$ is the transmission power. $P_{\mathrm{RF}}$ is the power consumption of each RF chain, and $N_{\mathrm{RF}}$ is the number of the RF chains, where $N_{\mathrm{RF}}=N$ for the fully digital structure and $N_{\mathrm{RF}}=M$ for the hybrid structure. $P_{\mathrm{PS}}$ is the power consumption of each PS, and $N_{\mathrm{RF}}$ is the number of the PSs, where $N_{\mathrm{PS}}=0$ for the fully digital structure and $N_{\mathrm{PS}}=MN$ for the hybrid structure. In the simulations, we select the typical parameter settings of $P=$1 W, $P_{\mathrm{RF}}=$250 mW, and $P_{\mathrm{PS}}=$1 mW\cite{Gao2016hyb}.

\begin{figure}
\begin{minipage}[t]{0.48\linewidth}
\centering
\includegraphics[width=\figwidth cm]{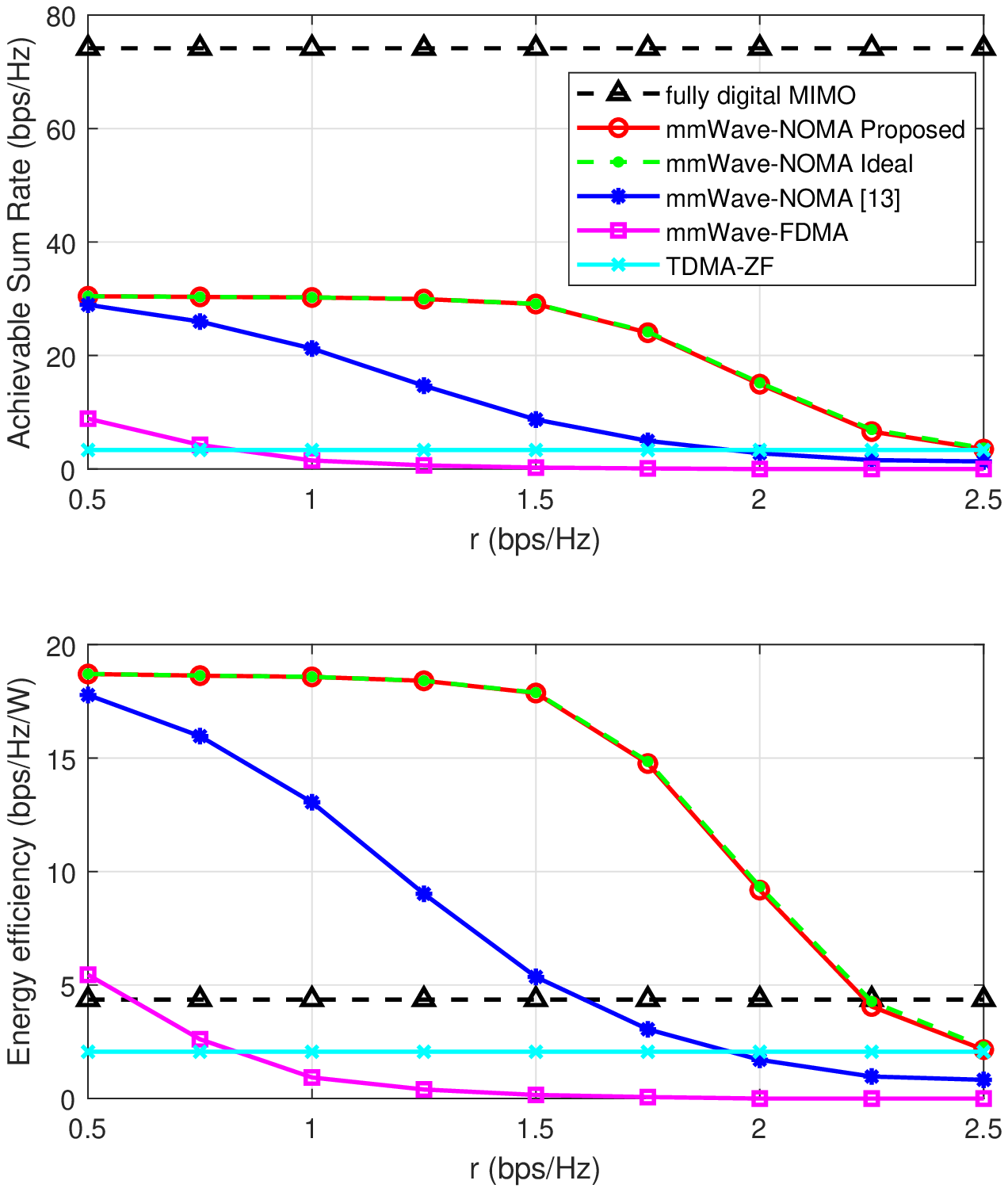}
  \caption{ASR/EE comparison between the mmWave-NOMA and mmWave-OMA systems with varying minimum rate constraint under the LOS channel model, where $N=64$, $M=2$, $K=6$, and $P/\sigma^2=30$ dB.}
  \label{Comp_r_LOS}
\end{minipage}
\begin{minipage}[t]{0.04\linewidth}
\end{minipage}
\begin{minipage}[t]{0.48\linewidth}
\centering
\includegraphics[width=\figwidth cm]{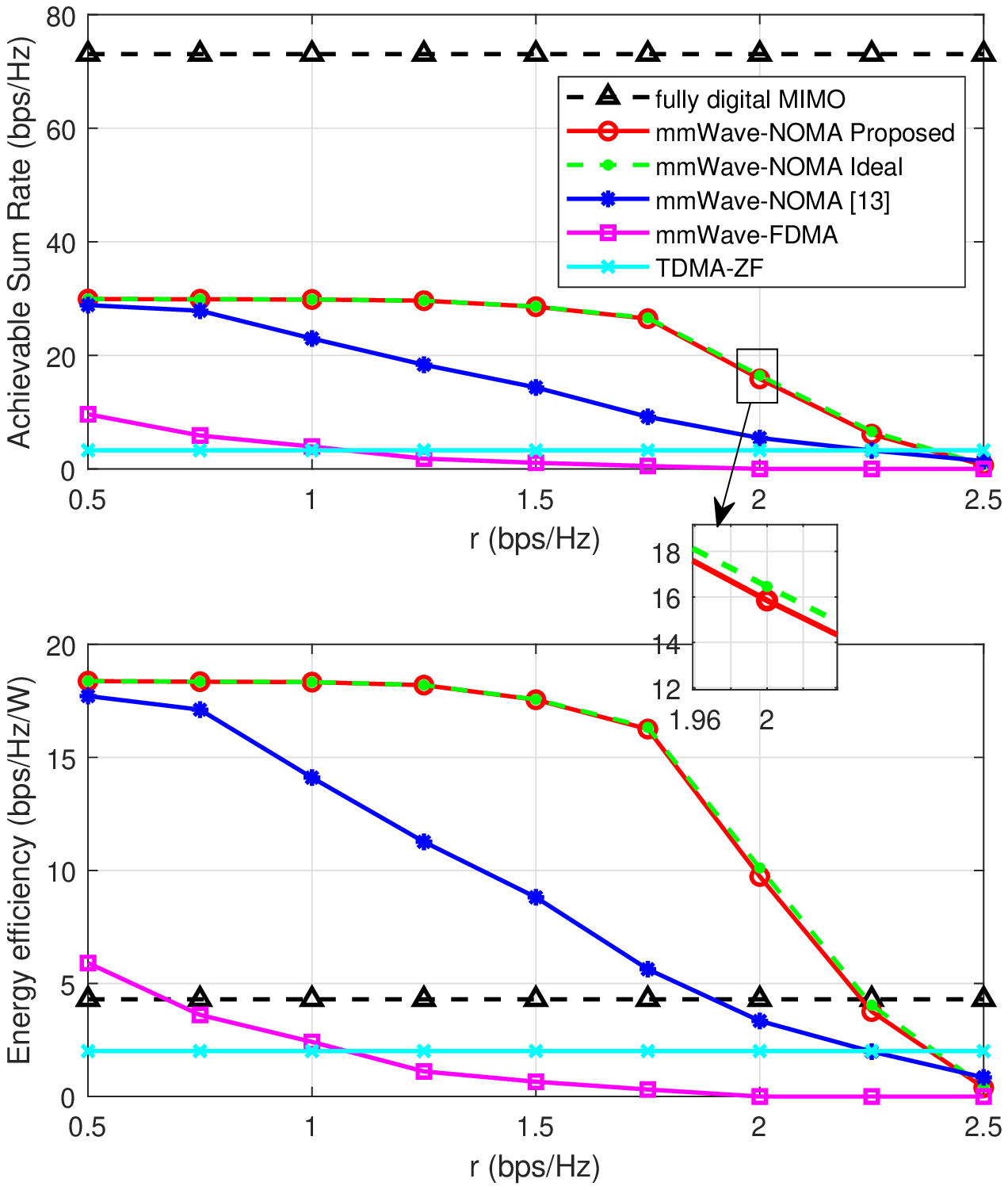}
  \caption{ASR/EE comparison between the mmWave-NOMA and mmWave-OMA systems with varying minimum rate constraint under the NLOS channel model, where $N=64$, $M=2$, $K=6$, and $P/\sigma^2=30$ dB.}
  \label{Comp_r_NLOS}
\end{minipage}
\end{figure}
Figs. \ref{Comp_r_LOS} and \ref{Comp_r_NLOS} show the ASR and EE comparisons between the proposed mmWave-NOMA approach, the mmWave-NOMA scheme in \cite{Dai2018MIMONOMA}, mmWave-OMA and fully digital MIMO with varying minimum rate constraint under the LOS channel and the NLOS channel models, respectively. The minimum rate constraints for all the users are equal to $r$. Clearly, the performance of the proposed mmWave-NOMA system is distinctly better than that of the mmWave-OMA system, TDMA, and the solution of mmWave-NOMA in \cite{Dai2018MIMONOMA}. Particularly, when the minimum rate constraint $r$ ranges from 1 to 2 bps/Hz, the ASR of the proposed approach is nearly 10 bps/Hz larger than that of the scheme in \cite{Dai2018MIMONOMA}. The reason is as follows. When $r$ is small, according to the NOMA principle, more beam gains and power can be allocated to the user with the highest channel gain in each group \cite{Zhu2018NOMAPSO}. Thus, the beamforming scheme in \cite{Dai2018MIMONOMA} is effective, where the beam in analog domain is steering to the first user in each group. When $r$ becomes larger, the users with worse channel conditions can only be served by the sidelobe of the beam in \cite{Dai2018MIMONOMA}. In contrast, the proposed solution in this paper can allocate more beam gains in analog domain to the users with worse channel conditions in each group. Thus, the proposed approach outperforms the scheme in \cite{Dai2018MIMONOMA}. However, when $r$ is large, there may exist some channel realizations in which the minimum rate constraint cannot be satisfied. In such a case, the ASR is set to be zero. This operation is also adopted in the scheme of \cite{Dai2018MIMONOMA}, which ensures the fairness of the comparison between the two methods. Therefore, the ASR tends to be zero for both of the two schemes, when $r$ is sufficiently large. Since the average ASR of the proposed scheme is larger than that of the scheme in \cite{Dai2018MIMONOMA}, it can be concluded that our method can find a better solution and achieve a higher feasibility. Besides, we can also find that the ASR of the proposed approach is close to the ideal case, which indicates that the inter-group interference is small by using the proposed user grouping and HBF schemes and has little influence on the ASR. This result also verifies that the approximation of neglecting the inter-group interference when optimizing the intra-GPA is reasonable. We have also provided an enlarged view of the ASR curve in Fig. \ref{Comp_r_NLOS}, it can be seen that there is a small gap between the ideal curve and the designed curve, which is caused by the inter-group interference. The performance gap is no more than 0.5 bps/Hz, which is very small compared with the total ASR. In the two figures, we can also find that, although the ASR of the fully digital MIMO structure is higher than that of both the mmWave-NOMA and mmWave-OMA, the EE of the fully digital MIMO structure is low compared with the HBF structure. Particularly, the EE of the proposed mmWave-NOMA scheme can achieve nearly fourfold EE compared with the fully digital MIMO structure when the minimal rate constraint is no more than 1.5 bps/Hz.

\begin{figure}
\begin{minipage}[t]{0.48\linewidth}
\centering
\includegraphics[width=\figwidth cm]{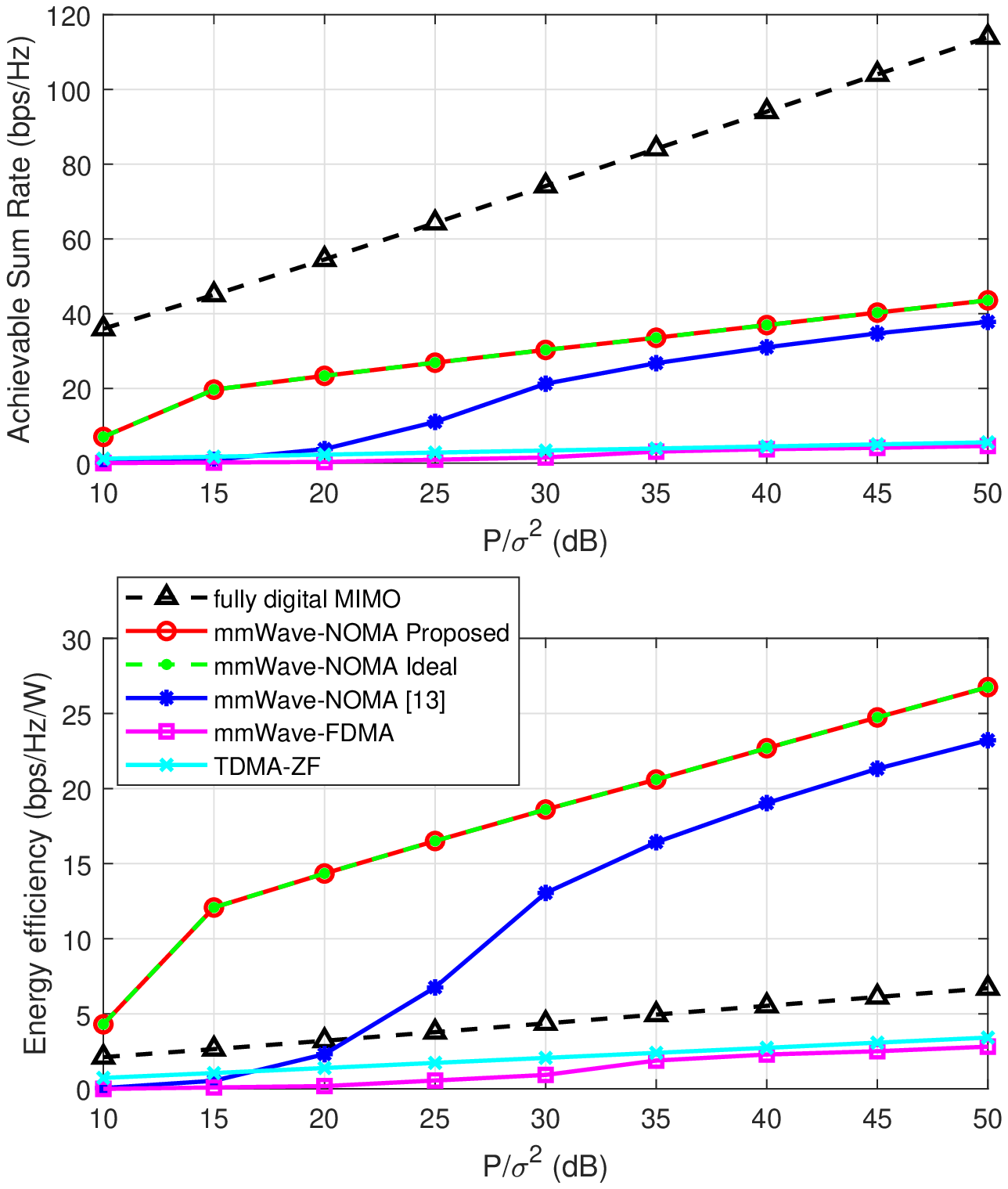}
  \caption{ASR/EE comparison between the mmWave-NOMA and mmWave-OMA systems with varying total power to noise ratio under the LOS channel model, where $N=64$, $M=2$, $K=6$, and $r_{k}=1$ bps/Hz.}
  \label{Comp_P_LOS}
\end{minipage}
\begin{minipage}[t]{0.04\linewidth}
\end{minipage}
\begin{minipage}[t]{0.48\linewidth}
\centering
\includegraphics[width=\figwidth cm]{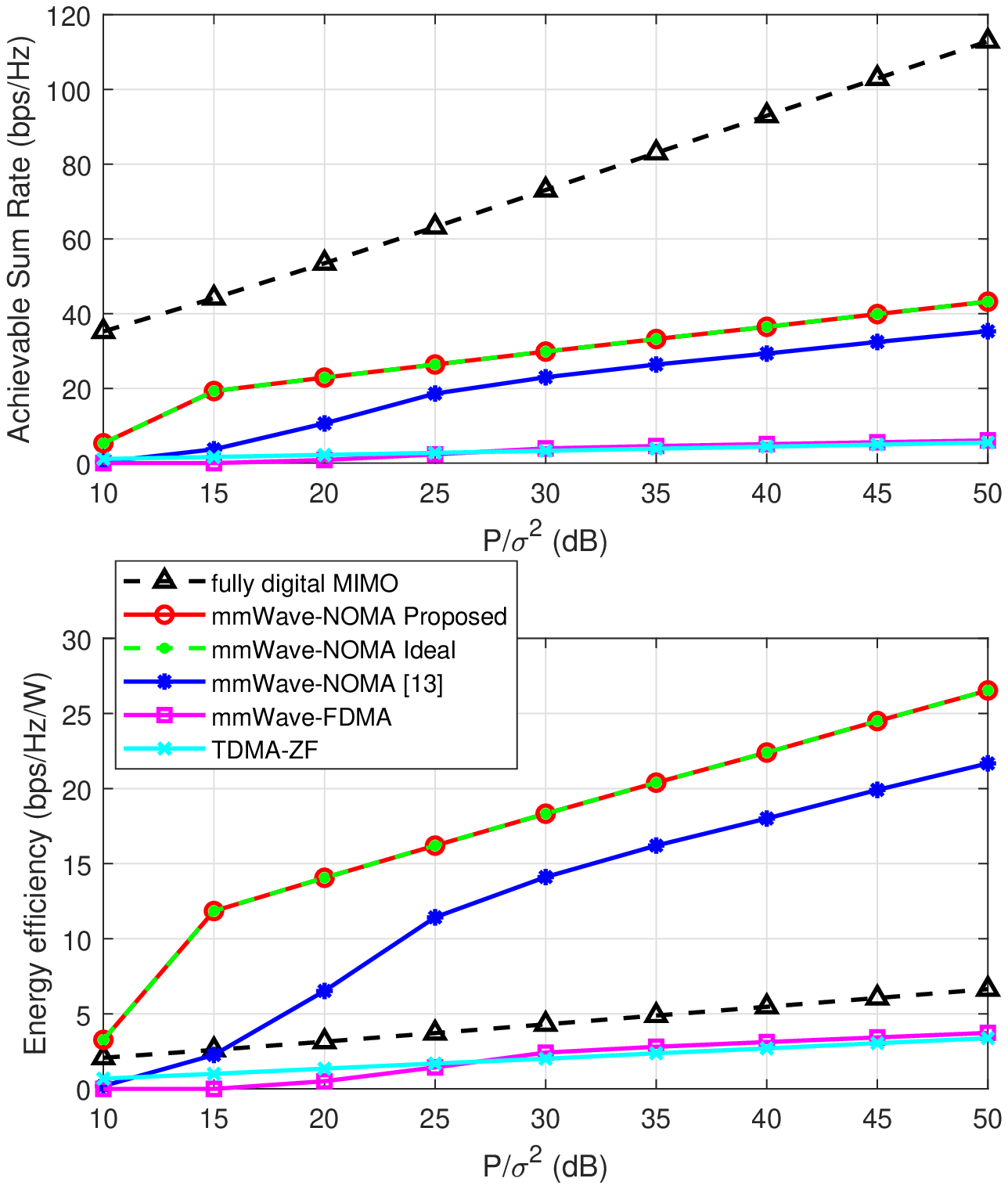}
  \caption{ASR/EE comparison between the mmWave-NOMA and mmWave-OMA systems with varying total power to noise ratio under the NLOS channel model, where $N=64$, $M=2$, $K=6$, and $r_{k}=1$ bps/Hz.}
  \label{Comp_P_NLOS}
\end{minipage}
\end{figure}

Figs. \ref{Comp_P_LOS} and \ref{Comp_P_NLOS} compare the ASRs/EEs between the proposed mmWave-NOMA approach, the mmWave-NOMA scheme in \cite{Dai2018MIMONOMA}, mmWave-OMA and fully digital MIMO with varying total power to noise ratio under the LOS channel and the NLOS channel models, respectively. From the two figures, we can find again that the proposed mmWave-NOMA approach can achieve a higher ASR than that of mmWave-NOMA in \cite{Dai2018MIMONOMA}, as well as the mmWave-OMA system. Particularly, when $P/\sigma^2$ is low, i.e., the mmWave-NOMA system is power limited, the superiority of the proposed algorithm is more conspicuous compared with the approach in \cite{Dai2018MIMONOMA}. When $P/\sigma^2$ is larger than 35 dB, the performance gap between the proposed solution and the solution in \cite{Dai2018MIMONOMA} stabilises around 5 bps/Hz in Fig. \ref{Comp_P_LOS}, while the performance gap stabilises around 7.5 bps/Hz in Fig. \ref{Comp_P_NLOS}. From the two figures, we can find again that the EE of the proposed mmWave-NOMA scheme with a HBF structure is larger than that of the fully digital MIMO structure, as well as larger than the EE of mmWave-OMA. When $P/\sigma^2$ becomes large, the curves of the EE for different schemes all tend to be linear, and the increasing velocity, i.e., the slope of the EE curve, for mmWave-NOMA is larger than that for both fully digital MIMO and mmWave-OMA.

\begin{figure}
\begin{minipage}[t]{0.48\linewidth}
\centering
\includegraphics[width=\figwidth cm]{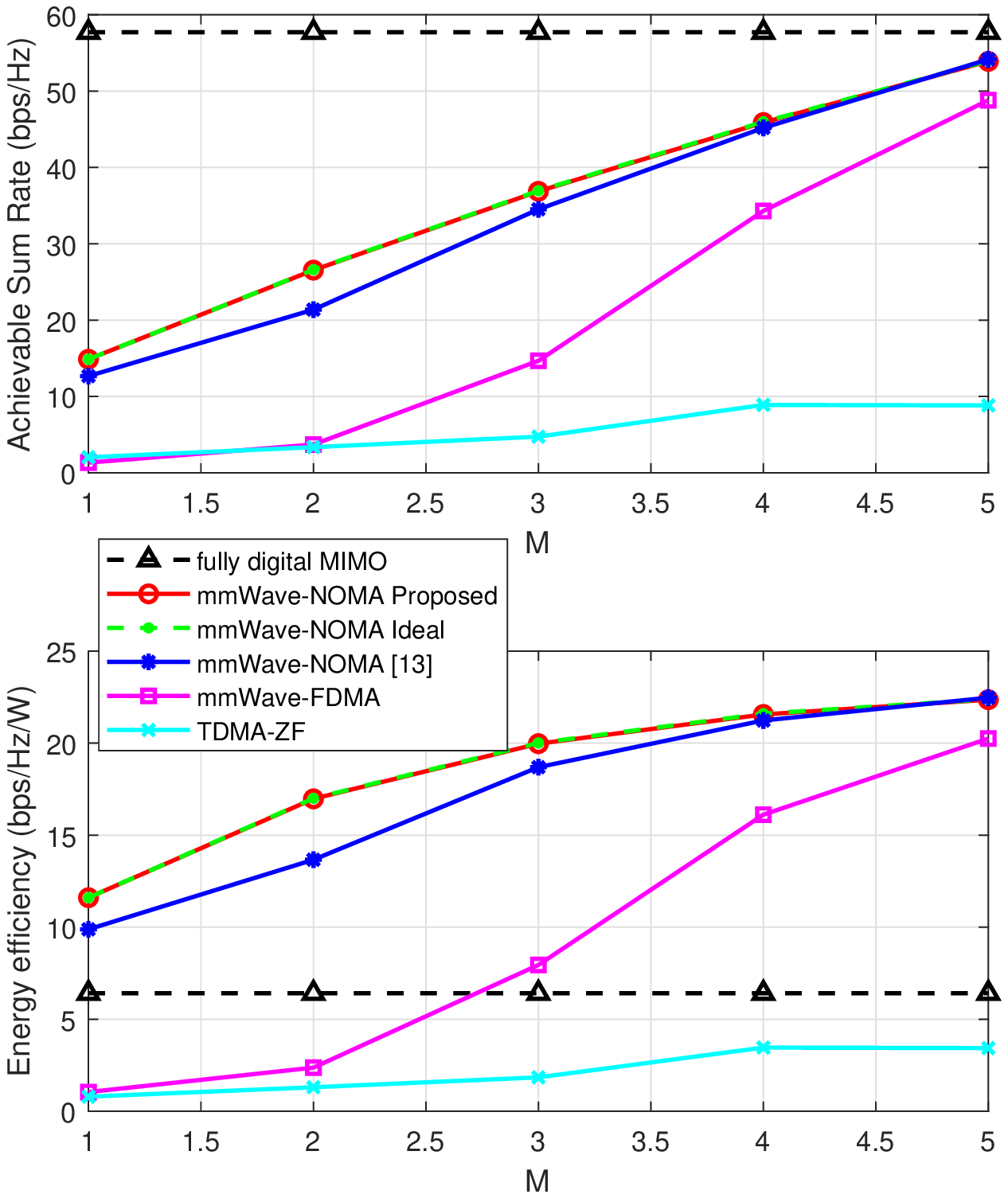}
  \caption{ASR/EE comparison between the mmWave-NOMA and mmWave-OMA systems with varying number of RF chains under the LOS channel model, where $N=16$, $K=6$, $r_{k}=1$ bps/Hz and $P/\sigma^2=30$ dB.}
  \label{Comp_RF_LOS}
\end{minipage}
\begin{minipage}[t]{0.04\linewidth}
\end{minipage}
\begin{minipage}[t]{0.48\linewidth}
\centering
\includegraphics[width=\figwidth cm]{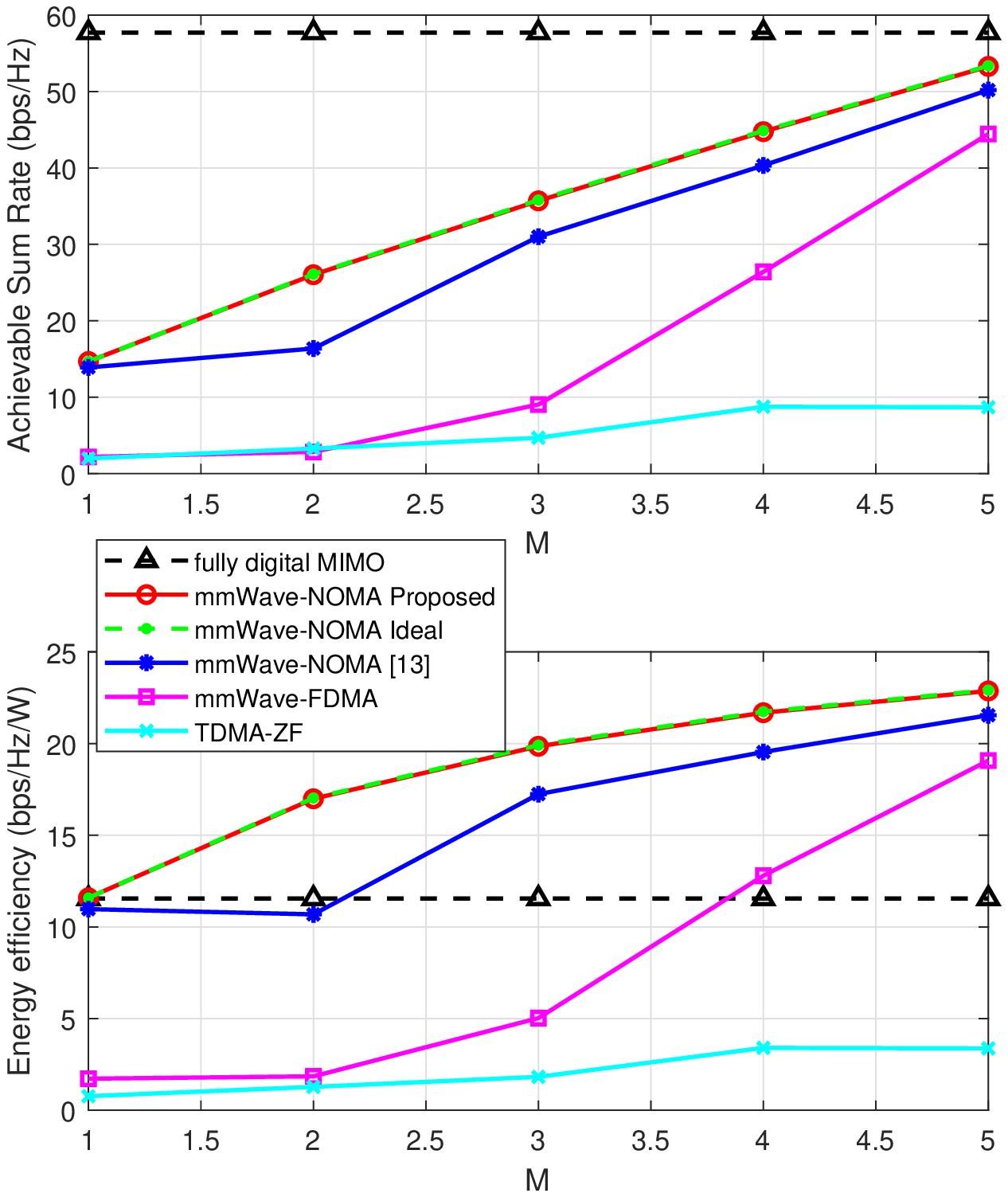}
  \caption{ASR/EE comparison between the mmWave-NOMA and mmWave-OMA systems with varying number of RF chains under the NLOS channel model, where $N=16$, $K=6$, $r_{k}=1$ bps/Hz and $P/\sigma^2=30$ dB.}
  \label{Comp_RF_NLOS}
\end{minipage}
\end{figure}

Figs. \ref{Comp_RF_LOS} and \ref{Comp_RF_NLOS} compare the ASRs/EEs between mmWave-NOMA and mmWave-OMA systems with varying number of RF chains under the LOS channel and the NLOS channel models, respectively. It can be observed that the proposed mmWave-NOMA approach outperforms the mmWave-OMA. In Fig. \ref{Comp_RF_LOS}, when the number of RF chains is no larger than 4, the ASR of the proposed approach is larger than that of mmWave-NOMA in \cite{Dai2018MIMONOMA}. When the number of RF chains is 5, the scheme in \cite{Dai2018MIMONOMA} behaves slightly better than the proposed scheme, and both of them are close to the performance of the fully digital structure. The reason is that the total number of users is 6 in Fig. \ref{Comp_RF_LOS}. When the number of RF chains becomes larger, i.e., approximately equal to the number of users, the number of users in each group is usually one. Thus, the beamforming scheme in \cite{Dai2018MIMONOMA} is more effective, where the analog beams steer to the first user in each group. Moreover, the proposed solution always outperforms the mmWave-NOMA scheme in \cite{Dai2018MIMONOMA} in Fig. \ref{Comp_RF_NLOS}. Comparing the two figures, we can find that the ASRs of the proposed approach are almost not influenced by the channel models. In contrast, the ASRs of the mmWave-NOMA scheme in \cite{Dai2018MIMONOMA} under the NLOS channel model is lower than that under the LOS channel model. The results indicate that the proposed approach is more robust against the channel model. We can also find that the EE of the proposed mmWave-NOMA scheme increases for the number of the RF chains, and it is significantly larger than the EE of the fully digital MIMO structure.

\section{Conclusion}
In this paper, we investigated the application of NOMA in mmWave communications. Particularly, we considered downlink transmission with HBF structure. First, we proposed the K-means based user grouping algorithm according to the channel correlations of the multiple users. Whereafter, a joint hybrid beamforming and power allocation problem was formulated to maximize the ASR, subject to a minimum rate constraint for each user. To solve this non-convex problem with high-dimensional variables, we first obtained a sub-optimal solution of power allocation under arbitrary fixed HBF, where the intra-GPA and inter-GPA sub-problems are solved, respectively. Then, given an arbitrary fixed ABF, we utilized the approximately zero-forcing method to design the DBF matrix to minimize the inter-group interference. Finally, the ABF problem with the CM constraint was solved by using the proposed BC-PSO algorithm. Simulation results showed that the proposed mmWave-NOMA scheme outperforms mmWave-OMA in terms of ASR, and the proposed mmWave-NOMA scheme with the HBF structure is more energy efficient compared with the fully digital MIMO structure. The proposed solution for mmWave-NOMA, including user grouping, joint power allocation and HBF, can achieve a better performance in terms of ASR and EE compared with the benchmark scheme, with the expending of a higher computational complexity.

\appendices
\section{Proof of Lemma 1}

It is obvious that $f(\{P_{m}\})$ is increasing for $P_{m} ~(1\leq m \leq M)$. Thus, the optimal solution always satisfies $\sum\limits_{m=1}^{M} P_{m} = P$. Then, Problem \eqref{eq_problem5} without the constraint $C_1$ can be solved by Lagrange Multiplier Method, where the KKT equation set is
\begin{equation}
\left\{\begin{aligned}
&\frac{\partial f}{\partial P_{m}} = \lambda,~(1\leq m \leq M)\\
&\sum\limits_{m=1}^{M} P_{m} = P.
\end{aligned}
\right.\end{equation}

Solve the equation sets above and we can obtain the optimal solution of Problem \eqref{eq_problem5} as shown in \eqref{inter_power}.

\section{Proof of Lemma 2}
We prove Lemma 2 by using contradiction. Denote the optimal solution of Problem \eqref{eq_problem5} is $\{P_{m}^{\circ}\}$. Assume there exists an index $m_{1}~(m_{1} \in \mathcal{U})$ which satisfies $P_{m_1}^{\circ}>\frac{\eta_{m_1,1}-b_{m_1}}{k_{m_1}}$. Since $m_{1} \in \mathcal{U}$, we have $P_{m_1}^{\circ}>\frac{\eta_{m_1,1}-b_{m_1}}{k_{m_1}}>P_{m_1}^{\star}$. Then, there always exists another index $m_{2}~(m_{2}\neq m_{1})$ which satisfies $P_{m_2}^{\circ}<P_{m}^{\star}$, because $\sum\limits_{m=1}^{M} P_{m}^{\circ}\leq P=\sum\limits_{m=1}^{M}P_{m}^{\star}$. Consider the power allocation of
\begin{equation}
\left\{
\begin{aligned}
&P_{m_1}^{\prime}=P_{m_1}^{\circ}-\epsilon \\
&P_{m_2}^{\prime}=P_{m_2}^{\circ}+\epsilon \\
&P_{m}^{\prime}=P_{m}^{\circ}, ~ m \neq m_1,m_2,
\end{aligned}
\right.
\end{equation}
where $\epsilon$ is a nonnegative and small number.

The partial derivative of the objective function is
\begin{equation}\label{pdiff}
\frac{\partial f}{\partial P_{i}}=\frac{1}{\ln 2}\frac{k_{i}}{ (k_{i}P_{i}+b_{i}+1)}, ~1\leq i\leq M,
\end{equation}
which is a monotone decreasing function of $P_{i}$. Since $P_{m_1}^{\circ}>P_{m_1}^{\star}$ and $P_{m_2}^{\circ}<P_{m_2}^{\star}$ , we have
\begin{equation}
\begin{aligned}
&\frac{\partial f}{\partial P_{m_1}}|\{P_{m}=P_{m}^{\circ}\}<\frac{\partial f}{\partial P_{m_1}}|\{P_{m}=P_{m}^{\star}\} \\
&\frac{\partial f}{\partial P_{m_2}}|\{P_{m}=P_{m}^{\circ}\}>\frac{\partial f}{\partial P_{m_2}}|\{P_{m}=P_{m}^{\star}\}
\end{aligned}
\end{equation}

Define $g(\epsilon)=f (\{P_{m}^{\prime}\})-f (\{P_{m}^{\circ}\})$. It is easy to verify that $g(0)=0$. The derivative of the function $g(\epsilon)$ is
\begin{equation}
\begin{aligned}
\frac{d~g}{d~\epsilon}&= \frac{d~f (\{P_{m}^{\prime}\})}{d~\epsilon}= \frac{\partial f}{\partial P_{m_2}}|\{P_{m}=P_{m}^{\circ}\}-  \frac{\partial f}{\partial P_{m_1}}|\{P_{m}=P_{m}^{\circ}\}\\
&> \frac{\partial f}{\partial P_{m_2}}|\{P_{m}=P_{m}^{\star}\}- \frac{\partial f}{\partial P_{m_1}}|\{P_{m}=P_{m}^{\star}\}=0,
\end{aligned}
\end{equation}
which means that $g(\epsilon)$ is a monotone increasing function of $\epsilon$. We can select a sufficiently small $\epsilon$ which satisfies $P_{m_1}^{\prime}>P_{m_1}^{\star}$, $P_{m_2}^{\prime}<P_{m_2}^{\star}$ and $g(\epsilon)=f (\{P_{m}^{\prime}\})-f (\{P_{m}^{\circ}\})>0$. In other words, $\{P_{m}^{\prime}\}$ is better than $\{P_{m}^{\circ}\}$. It contradicts to the assumption that $\{P_{m}^{\circ}\}$ is the optimal solution of Problem \eqref{eq_problem5}. Thus, we can conclude that for $\forall m \in \mathcal{U}$, the optimal solution of Problem \eqref{eq_problem5} should always satisfy $P_{m}^{\circ}=\frac{\eta_{m,1}-b_{m}}{k_{m}}$.

\bibliographystyle{IEEEtran} 
\bibliography{IEEEabrv,Xiao60GHz,Xiao5GnNOMA}

\end{document}